\documentclass[10pt,a4paper]{llncs}
\usepackage{graphicx}
\usepackage{subfigure}
\usepackage{amsmath}
\usepackage{amssymb}
\usepackage{amsmath}
\usepackage{multirow}
\usepackage{latexsym}
\usepackage{tikz}
\usepackage{url}
\usepackage{threeparttable}
\usepackage{hyperref}
\usepackage[a4paper]{geometry}
\renewcommand{\qed}{\hfill \ensuremath{\square}}
\newcommand{\qedTheorem}{\hfill \ensuremath{\blacksquare}}

\renewenvironment{proof}{
\vspace*{-\parskip}\noindent\textit{Proof.}}{$\qed$

\medskip
}

\newenvironment{proofTheorem}{
\vspace*{-\parskip}\noindent\textit{Proof.}}{$\qedTheorem$

\medskip
}

\def\hexdigit#1{\ifnum#1<10 \number#1\else
\ifnum#1=10 A\else\ifnum#1=11 B\else\ifnum#1=12 C\else
\ifnum#1=13 D\else\ifnum#1=14 E\else\ifnum#1=15 F\fi
\fi\fi\fi\fi\fi\fi}
\font\tenmsam=msam10

\font\sevenmsam=msam7
\font\fivemsam=msam5
\newfam\msafam
\textfont\msafam=\tenmsam
\scriptfont\msafam=\sevenmsam
\scriptscriptfont\msafam=\fivemsam

\font\tenmsbm=msbm10

\font\sevenmsbm=msbm7
\font\fivemsbm=msbm5
\newfam\msbfam
\textfont\msbfam=\tenmsbm
\scriptfont\msbfam=\sevenmsbm
\scriptscriptfont\msbfam=\fivemsbm
\def\mathbb{\fam=\msbfam\tenmsbm}

\mathchardef\N="0\hexdigit\msbfam4E
\mathchardef\R="0\hexdigit\msbfam52
\mathchardef\Z="0\hexdigit\msbfam5A
\mathchardef\F="0\hexdigit\msbfam46
\mathchardef\le="3\hexdigit\msafam36
\mathchardef\ge="3\hexdigit\msafam3E
\def   \q     {\mbox{\quad}}
\def   \mod   #1{\;({\rm mod}\;#1)}
\def   \Mod   #1{\;{\rm mod}\;#1}
\def   \ceil  #1{\left\lceil#1\right\rceil}
\def   \flr   #1{\left\lfloor#1\right\rfloor}
\def  \beas        {\begin{eqnarray*}}
\def  \eeas        {\end{eqnarray*}}

\begin{document}

\title{Efficient Proofs of Retrievability with Public Verifiability \\
for Dynamic Cloud Storage\thanks{A version of the paper with the same title has been published in 
IEEE Transactions on Cloud Computing (DOI: 10.1109/TCC.2017.2767584).
The final publication is available in \href{https://dx.doi.org/10.1109/TCC.2017.2767584}{IEEE Xplore}.}}

\author{Binanda Sengupta \and Sushmita Ruj}

\institute{%
Indian Statistical Institute, Kolkata, India\\
\email{\{binanda\_r,sush\}@isical.ac.in}}

\maketitle

\pagestyle{plain}

\begin{abstract}
Cloud service providers offer various facilities to their clients. The clients with limited resources opt
for some of these facilities. They can outsource their bulk data to the cloud server. The cloud
server maintains these data in lieu of monetary benefits. However, a malicious cloud server might delete
some of these data to save some space and offer this extra amount of storage to another client. Therefore, the client
might not retrieve her file (or some portions of it) as often as needed. Proofs of retrievability (POR)
provide an assurance to the client that the server is actually storing all of her data appropriately
and they can be retrieved at any point of time.
In a dynamic POR scheme, the client can update her data after she uploads them to the cloud server.
Moreover, in publicly verifiable POR schemes, the client can delegate her auditing task to some
third party specialized for this purpose. In this work, we exploit the \textit{homomorphic hashing} technique
to design a publicly verifiable dynamic POR scheme
that is more efficient (in terms of bandwidth required between the client and the server) than
the ``state-of-the-art'' publicly verifiable dynamic POR scheme. We also analyze security and performance
of our scheme.

\end{abstract}

\section{Introduction}\label{sec:intro}
In the age of cloud computing, cloud servers with adequate resources help their clients
by performing huge amount of computation or by storing large amount of data (say, in the order of terabytes)
on their behalf. In this setting, a client only has to download the result of the computation or has to read
(or update) the required portion of the outsourced data.
Several storage service providers like Amazon Simple Storage Service (S3), Dropbox, Google Drive and Microsoft Azure
provide storage outsourcing facility to their clients (data owners).
However, a cloud storage server can be malicious and delete some (less frequently accessed) part of the client's data in
order to save space. Secure cloud storage protocols (two-party protocols between the client and the server)
provide a cryptographic solution to this problem by ensuring that the client's data are stored untampered in
the cloud server.

In a secure cloud storage scheme, a client can remotely check the integrity of her data file outsourced
to the cloud server. A possible way to do this is to divide the data file into blocks and attach
an authenticator (or tag) to each of these blocks before the initial outsourcing. When the client wants
to check the integrity of her data, she downloads all the blocks of the file along with their tags from the server
and verifies them individually. However, this process is highly inefficient due to the large communication
bandwidth required between the client and the cloud server.

In order to resolve the issue mentioned above, a notion called \textit{proofs-of-storage} comes
into play where the client can audit her data file stored in the server without accessing
the whole file, and still, be able to detect an unwanted modification of the file
done by the malicious server.
The concept of \textit{provable data possession} (PDP) is introduced by Ateniese et al.~\cite{Ateniese_CCS}
where the client computes an authentication tag for each block of her data file
and uploads the file along with these authentication tags as stated earlier.
Later, the client can execute an \textit{audit} protocol and verify the integrity
of the data file by checking only a predefined number of randomly sampled blocks of the file.

Although efficient PDP schemes~\cite{Ateniese_CCS,Erway_TISSEC,Wang_TPDS,Wang_TC}
are available in the literature, they only provide the guarantee of retrievability
of \textit{almost all} blocks of the data file. We briefly mention some situations
where this guarantee is not sufficient. The data file may contain some sensitive
information (e.g., accounting information) any part of which the client does not want to
lose. On the other hand, a corruption in the compression table of an outsourced compressed file might
make the whole file unavailable. In a searchable symmetric encryption scheme~\cite{CurtmolaGKO06},
the client encrypts a database using a symmetric encryption scheme to form an index
(metadata for that database)~\cite{Goh03} and outsources the encrypted documents along with
the index to the server. The size of this index is very small compared to the encrypted
database itself. However, loss of the index completely defeats the purpose of the searchable
encryption scheme. In such circumstances, the client wants a stronger notion than PDP
which would guarantee that the server has stored the \textit{entire} file properly and
the client can retrieve \textit{all} of her data blocks at any point of time.

To address the issue mentioned above, Juels and Kaliski~\cite{JK_CCS} introduce
\textit{proofs of retrievability} (POR) where the data file outsourced to the server can be
retrieved in its entirety by the client. The underlying idea~\cite{SW_ACR} of
a POR scheme is to encode the original data file with an error-correcting code, authenticate
the blocks of the encoded file with tags and upload them on the storage server.
As in PDP schemes, the client can execute an audit protocol to check the integrity
of the outsourced data file. The use of error-correcting codes ensures that \textit{all}
data blocks of the file are retrievable.
Depending on the nature of the outsourced data, POR schemes are classified as:
POR schemes for \textit{static} data (static POR) and \textit{dynamic} data
(dynamic POR). For static data, the client cannot change her data after the initial outsourcing
(suitable mostly for backup or archival data). Dynamic data are more generic in that the client
can modify her data as often as needed.
The POR schemes are \textit{publicly verifiable} if audits can be performed by any third party auditor (TPA)
with the knowledge of public parameters only; they are \textit{privately verifiable} if only the client (data owner)
with some secret information can perform audits.

Designing an \textit{efficient} and publicly verifiable dynamic POR scheme is an important research problem due to its practical applications.
There are various efficiency parameters where the performance of a dynamic POR scheme might be improved.
Some of them include the communication bandwidth required to read or write a data block (or to perform
an audit), the client's storage, and the computational cost at the client's (or the server's) end.
On the other hand, the client often prefers to delegate the auditing task to a third party auditor (TPA)
who performs audits on the client's data and informs the client in case of any anomaly detected in
the server's behavior.

Shi et al.~\cite{Stefanov_CCS} propose two efficient dynamic POR schemes: one with
private verifiability and another with public verifiability.
In this work, we provide a construction of a publicly verifiable dynamic POR scheme that is more efficient (in terms of write and audit costs)
than the publicly verifiable scheme proposed by Shi et al.~\cite{Stefanov_CCS}.
Moreover, the public parameters used in the latter scheme need to be changed for \textit{each} write
operation performed on the client's data which is clearly an overhead as, in that case, these public parameters are to
be validated and certified for every write.

\medskip
\noindent{\bf Our Contribution}\q
We summarize our contributions in this paper as follows.
\begin{itemize}
\item We construct a \textit{dynamic} proofs-of-retrievability (POR) scheme where the client outsources
her data file to the cloud server and she can update the content of the file later.
Our construction is based on the \textit{homomorphic hashing} technique.\smallskip

\item Our dynamic POR scheme offers \textit{public verifiability}, that is, the client can delegate the auditing task
to a third party auditor who performs audits on the client's behalf. \smallskip

\item We show that our scheme is secure in the sense that the client gets an assurance
that her data file stored by the server is authentic and up-to-date, and \textit{all}
the data blocks can be retrieved by the client as often as needed. \smallskip

\item We analyze the performance of our scheme and compare it with other existing dynamic POR
schemes (having private or public verifiability). \smallskip

\item Our publicly verifiable dynamic POR scheme enjoys more efficiency (in terms of communication bandwidths
required for a write and an audit) than the ``state-of-the-art'' publicly verifiable dynamic POR scheme~\cite{Stefanov_CCS}.
Moreover, unlike the latter scheme, there is no need to validate or certify the public parameters in our scheme
for every write operation as they are fixed since the initial setup phase.
\end{itemize}

\medskip
The rest of the paper is organized as follows.
Section~\ref{prelims} describes the preliminaries and background related to our work.
In Section~\ref{scs}, we survey the existing literature on secure cloud storage schemes.
Section~\ref{scheme} provides a detailed construction of our publicly verifiable dynamic POR scheme.
We analyze the security of our scheme in Section~\ref{security}.
Finally, in Section~\ref{performance_ana}, we discuss the performance
of our scheme and compare our scheme with other existing dynamic POR schemes
based on different parameters (shown in Table~\ref{tab:comparison_POR}). We also show that our scheme is more efficient
than the publicly verifiable dynamic POR scheme proposed by Shi et al.~\cite{Stefanov_CCS}.
In the concluding Section~\ref{sec:conclusion}, we summarize the work done in this paper.

\section{Preliminaries and Background}
\label{prelims}

\subsection{Notation}
We take $\lambda$ to be the security parameter.
An algorithm denoted by $\mathcal{A}(1^\lambda)$ is a probabilistic polynomial-time
algorithm when its running time is polynomial in $\lambda$ and its output $y$
is a random variable which depends on the internal coin tosses of $\mathcal{A}$.
If $\mathcal{A}$ is given access to an oracle $\mathcal{O}$, we denote $\mathcal{A}$
by $\mathcal{A}^{\mathcal{O}}$ also.
An element $a$ chosen uniformly at random from a set $S$
is denoted as $a\xleftarrow{R}S$.
A function $f:\N\rightarrow\R$ is called negligible in $\lambda$ if for
all positive integers $c$ and for all sufficiently large $\lambda$,
we have $f(\lambda)<\frac{1}{\lambda^c}$.

\subsection{Erasure Code}
\label{erasure_code}
A $(\tilde{m},\tilde{n},d)_\Sigma$-erasure code~\cite{ErasureCode_FAST05_Tutorial,ErasureCode_FAST12} is an error-correcting code~\cite{MWSloane77}
that comprises an encoding algorithm Enc: $\Sigma^{\tilde{n}}\rightarrow\Sigma^{\tilde{m}}$
(encodes a message consisting of $\tilde{n}$ symbols into a longer codeword consisting of $\tilde{m}$ symbols) and
a decoding algorithm Dec: $\Sigma^{\tilde{m}}\rightarrow\Sigma^{\tilde{n}}$ (decodes a codeword to a message),
where $\Sigma$ is a finite alphabet and $d$ is the minimum distance
(Hamming distance between any two codewords is at least $d$) of the code.
The quantity $\frac{\tilde{n}}{\tilde{m}}$ is called the rate of the code.
A $(\tilde{m},\tilde{n},d)_\Sigma$-erasure code can tolerate up to $d-1$ erasures.
If $d=\tilde{m}-\tilde{n}+1$, we call the code a maximum distance separable (MDS) code.
For an MDS code, the original message can be reconstructed from any $\tilde{n}$
out of $\tilde{m}$ symbols of the codeword. Reed-Solomon codes~\cite{RSCode}
and their extensions are examples of non-trivial MDS codes.

\subsection{Merkle Hash Tree}
\label{MHT}
A Merkle hash tree~\cite{Merkle_CR} is a binary tree where each leaf-node stores a data item.
The label of each leaf-node is the data item stored in the node itself.
A collision-resistant hash function $h_{CR}$ is used to label the intermediate nodes of the tree.
Each of the outputs of $h_{CR}$ on different inputs is a binary string of length $O(\lambda)$.
The label of a intermediate node $v$ is the output of $h_{CR}$ computed on the labels of the children
nodes of $v$.
A Merkle hash tree is used as a standard tool for efficient memory-checking.
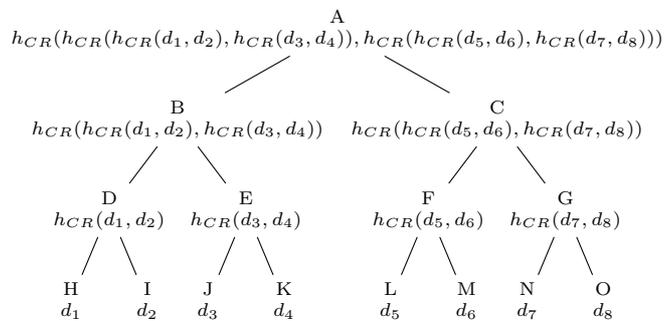
\begin{figure}[bp]
\scriptsize
\centering
\begin{tikzpicture}[level distance=1.2cm,
                      level 1/.style={sibling distance=4.2cm}, 
                      level 2/.style={sibling distance=1.8cm}, 
                      level 3/.style={sibling distance=1cm}, 
                      every text node part/.style={align=center}]
                    \node {A \\ $h_{CR}(h_{CR}(h_{CR}(d_1, d_2), h_{CR}(d_3, d_4)), h_{CR}(h_{CR}(d_5, d_6), h_{CR}(d_7, d_8)))$}
                        child {node {B \\ $h_{CR}(h_{CR}(d_1, d_2), h_{CR}(d_3, d_4))$}
                                child {node {D \\ $h_{CR}(d_1, d_2)$}
                                    child {node {H \\ $d_1$}}
                                    child {node {I \\ $d_2$ }}
                                }
                                child {node {E \\ $h_{CR}(d_3, d_4)$}
                                    child {node {J \\ $d_3$}}
                                    child {node {K \\ $d_4$}}
                                }
                        }
                        child {node {C \\ $h_{CR}(h_{CR}(d_5, d_6), h_{CR}(d_7, d_8))$}
                                child {node {F \\ $h_{CR}(d_5, d_6)$}
                                    child {node {L \\ $d_5$}}
                                    child {node {M \\ $d_6$}}
                                }
                                child {node {G \\ $h_{CR}(d_7, d_8)$}
                                    child {node {N \\ $d_7$}}
                                    child {node {O \\ $d_8$}}
                                }
                        };
                    \end{tikzpicture}

\caption{A Merkle hash tree containing data items $\{d_1,d_2,\ldots ,d_8\}$.}
\label{fig:MHT}
\end{figure}
Fig.~\ref{fig:MHT} shows a Merkle hash tree containing the data items $\{d_1,d_2,\ldots,d_8\}$ 
stored at the leaf-nodes. Consequently, the labels of the intermediate nodes are computed
using the hash function $h_{CR}$. The hash value of the node $A$ is the \textit{root-digest}.
The proof showing that a data item $d$ is present in the tree consists of the data item $d$ and
the labels of the nodes along the \textit{associated path} (the sequence of siblings of the node
containing the data item $d$). For example, a proof showing that $d_3$ is present in the tree
consists of $\{d_3,(d_4,l_D,l_C)\}$, where $d_4,l_D$ and $l_C$ are the labels of the nodes
$K,D$ and $C$, respectively. Given such a proof, a verifier computes the hash value of the root.
The verifier outputs \texttt{accept} if the computed hash value matches with the root-digest;
it outputs \texttt{reject}, otherwise. The size of a proof is logarithmic in the number of data items
stored in the leaf-nodes of the tree.

Due to the collision-resistance property of $h_{CR}$, it is infeasible (except with some
probability negligible in the security parameter $\lambda$) to add or modify a data item
in the Merkle hash tree without changing its root-digest.

\subsection{Digital Signature Scheme}
\label{dig_sig}
Diffie and Hellman introduce the public-key cryptography and the notion
of digital signatures in their seminal paper ``New Directions in
Cryptography''~\cite{DH_ITIT}. Rivest et al.~propose the first digital signature scheme
based on the RSA assumption\cite{RSA_CACM}.
Boneh et al.~\cite{BLS_JOC} introduce the first signature scheme
where the signatures are short (e.g., such a signature of size 160 bits provides the security comparable to
that of a 1024-bit RSA signature).
The DSA (Digital Signature Algorithm)~\cite{DSA} and ECDSA (Elliptic Curve Digital Signature Algorithm)~\cite{ECDSA}
signature schemes (variants of the ElGamal signature scheme~\cite{ElGamal_CR}) are widely used in practice.

A digital signature scheme
consists of the following polynomial-time algorithms:
a key generation algorithm KeyGen, a signing algorithm Sign
and a verification algorithm Verify. KeyGen takes as input the security parameter
$\lambda$ and outputs a pair of keys $(psk,ssk)$, where $ssk$ is the
secret key and $psk$ is the corresponding public verification key. Algorithm
Sign takes a message $m$ from the message space $\mathcal{M}$
and the secret key $ssk$ as input
and outputs a signature $\sigma$.
Algorithm Verify takes as input the public key $psk$, a message $m$
and a signature $\sigma$, and outputs \texttt{accept} or \texttt{reject}
depending upon whether the signature is valid or not.
Any of these algorithms can be probabilistic in nature.
The correctness and security (existential unforgeability under adaptive
chosen message attacks~\cite{GMR_ACM}) of a digital signature scheme are described as follows.
\begin{enumerate}
 \item \textit{Correctness}:\q Algorithm Verify always accepts a signature generated by an honest signer, that is,
	 \begin{align*}
	  \Pr[\text{Verify}_{psk}(m,\text{Sign}(ssk,m))=\texttt{accept}]=1.
	 \end{align*}
 \item \textit{Security}:\q Let Sign$_{ssk}(\cdot)$ be the signing oracle and $\mathcal{A}$
	 be any probabilistic polynomial-time adversary with an oracle access to Sign$_{ssk}(\cdot)$.
	 The adversary $\mathcal{A}$ adaptively makes polynomial
	 number of sign queries to Sign$_{ssk}(\cdot)$ for different messages and gets back the signatures
	 on those messages. The signature scheme is secure if $\mathcal{A}$ cannot produce,
	 except with some probability negligible in $\lambda$,
	 a valid signature on a message not queried previously, that is, for
	 any probabilistic polynomial-time adversary $\mathcal{A}^{\text{Sign}_{ssk}(\cdot)}$,
	 the following probability
	 \beas
	  \Pr[(m,\sigma)\leftarrow\mathcal{A}^{\text{Sign}_{ssk}(\cdot)}(psk):
	  {m\not\in Q_{s}} \wedge \text{Verify}_{psk}(m,\sigma)=\texttt{accept}]
	 \eeas
	 is negligible in $\lambda$, where $Q_{s}$ is the set of sign queries made by
	 $\mathcal{A}$ to $\text{Sign}_{ssk}(\cdot)$.
	 The probability is taken over the internal coin tosses of $\mathcal{A}$.
\end{enumerate}

\subsection{Discrete Logarithm Assumption}
\label{disLog}
The discrete logarithm problem~\cite{Dlog_McCurley,Dlog_Bellare} over a multiplicative group $G_q=\langle g \rangle$ of prime order $q$ and generated by $g$
is defined as follows.
\begin{definition}[Discrete Logarithm Problem]
Given $y\in G_q$, the discrete logarithm problem over $G_q$ is to compute $x\in\Z_q$
such that $y=g^x$.
\end{definition}
The discrete logarithm assumption over $G_q$ says that, for any probabilistic polynomial-time
adversary $\mathcal{A}(1^\lambda)$, the probability
\beas
\Pr_{x\xleftarrow{R}\Z_q}[x\leftarrow\mathcal{A}(y): y=g^x]
\eeas
is negligible in $\lambda$, where the probability is taken over the internal
coin tosses of $\mathcal{A}$ and the random choice of $x$.

\begin{figure}[t]
\centering
\includegraphics[width=.4755\textwidth]{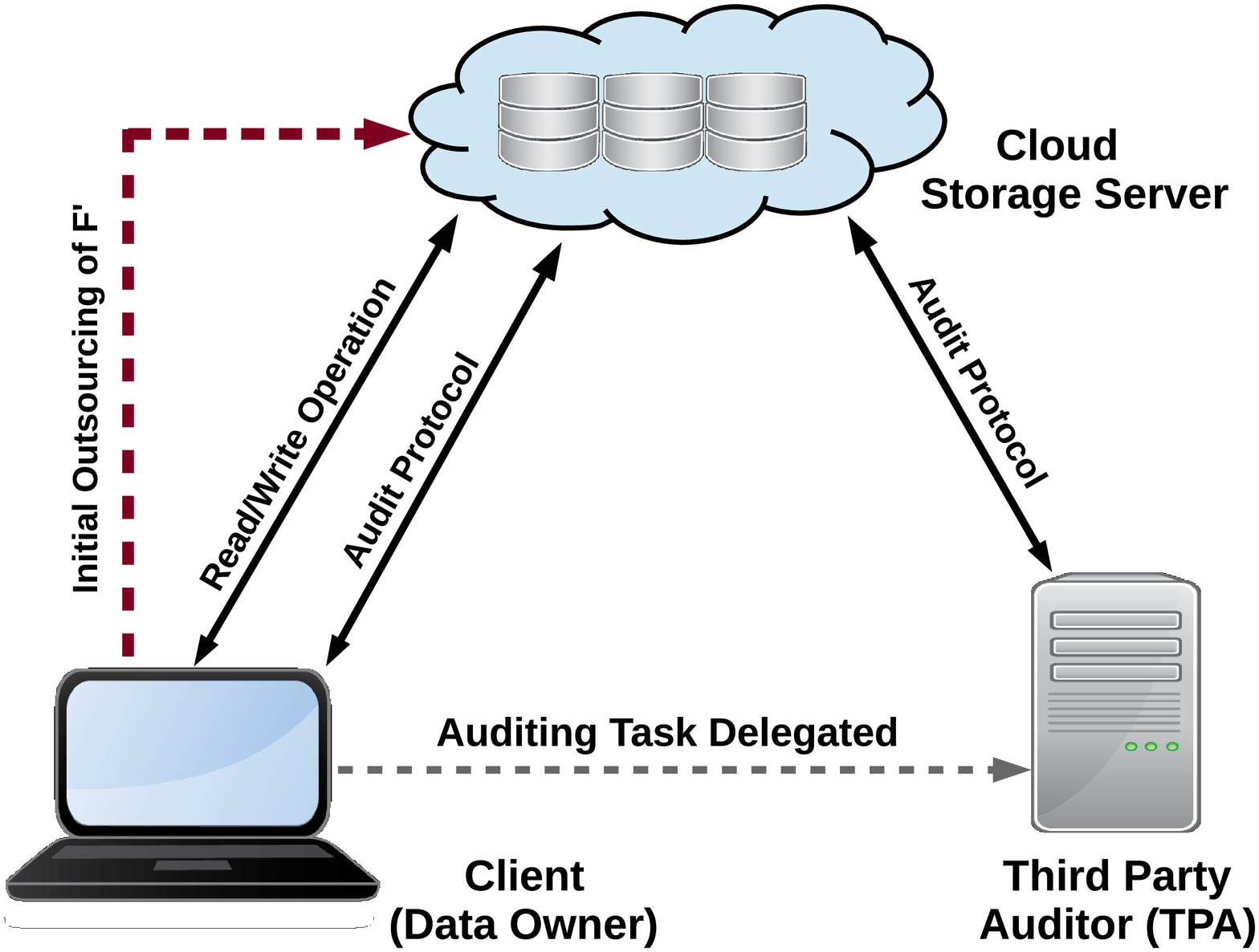}
\caption{The entities involved in a dynamic POR scheme.
	 The client (data owner) processes the data file $F$ to form another file $F'$ and outsources it to the cloud storage server.
	 She can later read or write the outsourced data. For a privately verifiable scheme,
	 the client performs audits on her data.
	 For a publicly verifiable scheme, she can delegate the auditing task to a third party auditor who performs audits
	 on behalf of the client.}
	 \label{fig:dpor}
\end{figure}

\subsection{Dynamic Proofs of Retrievability}
\label{dpor_def}
We define a proofs-of-retrievability scheme for \textit{dynamic} data as follows~\cite{Wichs_ORAM,Stefanov_CCS}.
\begin{definition}[Dynamic POR]
A dynamic POR scheme consists of the following protocols between two stateful parties: a client (data owner) and a server.
\begin{itemize}
\item \emph{Init($1^\lambda,n,\beta,F$):} This protocol associates a random file-identifier \emph{\texttt{fid}}
to the data file $F$ consisting of $n$ data blocks each of $\beta$ bits, and it outputs the client
state $state_C$ and another file $F'$ to be stored by the server.

\item \emph{Read($i,F',state_C,\texttt{fid}$):} This protocol outputs the data block at the $i$-th location of
the current state of the file or \emph{\texttt{abort}}.

\item \emph{Write($i,\texttt{updtype},B,F',state_C,\texttt{fid}$):} This protocol inserts the block $B$ after the $i$-th block of the file
or sets $i$-th block of the file to $B$ or deletes the $i$-th block of the file ($B$ is \emph{\texttt{null}}
in this case) based on the value of the variable \emph{\texttt{updtype}}. 
It outputs updated $(\tilde{F}',\widetilde{state}_C)$ or \emph{\texttt{abort}}.

\item \emph{Audit($F',state_C,\texttt{fid}$):} This protocol checks memory contents of the current state of the data file
and outputs 1 or 0.

\end{itemize}
\end{definition}
A dynamic POR scheme is \textit{privately verifiable} if only the client with some secret information can perform an audit,
that is, $state_C$ is secret. Otherwise, it is \textit{publicly verifiable}.
For a publicly verifiable dynamic POR scheme, a third party auditor (TPA) can audit
the client's data on behalf of the client who delegates her auditing task to the TPA. In this case, we use the term ``verifier''
to denote an auditor who can be the TPA or the client herself.
Fig.~\ref{fig:dpor} shows the entities involved in a dynamic POR scheme.
Security of a dynamic POR scheme is described
in Section~\ref{security}.

\subsection{Homomorphic Hash Function}
\label{hom_hash}
A homomorphic hash function $h: \F^m\rightarrow G_q$
(for a finite field $\F$ and a multiplicative group $G_q$ of prime order $q$)
is defined as a collision-resistant hash function satisfying the following two properties:
1) for vectors $\hbox{u},\hbox{v}\in\F^m$ and scalars $\alpha,\beta\in\F$,
it holds that $h(\alpha\hbox{u} + \beta\hbox{v}) = h(\hbox{u})^{\alpha}\cdot h(\hbox{v})^{\beta}$,
and 2) it is computationally hard to find vectors $\hbox{u},\hbox{v}\in\F^m$ ($\hbox{u}\not=\hbox{v}$)
such that $h(\hbox{u})=h(\hbox{v})$.

Krohn et al.~\cite{Krohn_SP} construct a homomorphic hash function in the context of content distribution.
The construction is similar to that proposed in incremental hashing schemes~\cite{IncCrypto_CR}.
Let $G_q$ be a multiplicative group of prime order $q$.
Let $m$ elements (generators) $g_1,g_2,\ldots,g_m$ be selected randomly from $G_q$.
Then, the homomorphic hash of a vector $\hbox{u}=[u_{1},u_{2},\ldots,u_{m}]\in \Z_q^m$ is defined as
  $h(\hbox{u})=\prod_{i=1}^{m}{g_i}^{u_i}$.
The hash function thus constructed is homomorphic, and
the collision-resistance property is derived from the discrete logarithm assumption over $G_q$.
We use this construction in our dynamic POR scheme to generate
authentication tags for data blocks (see Section~\ref{tag_gen}).

\section{Related Work}
\label{scs}

Ateniese et al.~\cite{Ateniese_CCS} introduce the notion of \textit{provable data possession} (PDP).
In a PDP scheme, the client computes an authentication tag (e.g., message authentication code~\cite{CBC_MAC}) for each block
of her data file and uploads the file along with the authentication tags. During an audit protocol,
the client samples a predefined number of random block-indices and sends them to the server
(\textit{challenge} phase). The cardinality of the challenge set is typically taken to be $O(\lambda)$,
where $\lambda$ is the security parameter. Depending upon the challenge, the server does some computations
over the stored data and sends a proof to the client (\textit{response} phase). Finally, the client checks
the integrity of her data based on this proof (\textit{verification} phase).
\textit{Almost all} data blocks can be retrieved from a (possibly malicious) server
passing an audit with a non-negligible probability.
Other PDP schemes include~\cite{Ateniese_SCOM,Erway_TISSEC,Wang_TPDS,Wang_TC,Curtmola_ICDCS}.

Juels and Kaliski~\cite{JK_CCS} introduce \textit{proofs of retrievability} (POR) for static data
(Naor and Rothblum~\cite{NR_JACM} give a similar idea for sublinear authenticators).
According to Shacham and Waters~\cite{SW_ACR}, the retrievability guarantee for \textit{all}
data blocks of the outsourced file can be achieved by encoding the original file with an
erasure code (see Section~\ref{erasure_code}) before authenticating (and uploading) the blocks
of the encoded file. Due to the redundancy added to the data blocks, the server has to delete
or modify a considerable number of blocks to actually delete or modify a data block which
makes it difficult for the server to pass an audit.

Following the work by Juels and Kaliski,
several POR schemes have been proposed for static data (\textit{static POR}) and dynamic data (\textit{dynamic POR}).
Shacham and Waters~\cite{SW_ACR} propose two POR schemes for static data (one with private verifiability and another with
public verifiability) where the response from the server is short.
Bowers et al.~\cite{Bowers_CCSW} propose a theoretical framework for designing POR schemes
and provide a prototype implementation of a variant of~\cite{JK_CCS}. In another work,
Bowers et al.~\cite{Bowers_HAIL} introduce HAIL (High-Availability and Integrity Layer),
a distributed POR setting where the client's data are disseminated across multiple servers.
Dodis et al.~\cite{Wichs_HA} introduce a notion called ``POR codes'' and show how POR schemes can be
instantiated based on these POR codes. They explore the connection between POR codes and the notion
of hardness amplification~\cite{Hard_Amp}.
Xu and Chang~\cite{Xu_ASIACCS} improve the privately verifiable scheme of~\cite{SW_ACR}
by making the communication bandwidth required for an audit to be $O(\lambda)$.
Armknecht et al.~\cite{Outpor_CCS} propose a POR scheme where any entity among the client (data owner),
the auditor and the cloud server can be malicious, and any two of them can collude as well.
The auditor performs two audits: one for the auditor itself and another on behalf
of the client. The challenge sets are generated using a public randomized algorithm derived
from the hash value of the latest block added to the Bitcoin block chain~\cite{Nakamoto08}.

A few dynamic POR schemes are there in the literature.
Stefanov et al.~\cite{IRIS} propose an authenticated file system called ``Iris''
that is highly scalable and resilient against a malicious cloud server.
Cash et al.~\cite{Wichs_ORAM} encode a small number of data blocks \textit{locally}
and hide the access pattern of the related (belonging to the same codeword) blocks
from the server using oblivious RAM (ORAM)~\cite{ORAM_JACM}. Due to the use of
expensive ORAM primitives, this scheme is inefficient.
Shi et al.~\cite{Stefanov_CCS} propose two practical dynamic POR schemes which reduce
the cost of computation as well as the communication bandwidth required to execute
the protocols involved.
Chandran et al.~\cite{Bhavana_TCC} introduce the notion of ``locally updatable and
locally decodable codes'' and propose an efficient dynamic POR scheme by
applying the techniques used in the construction of such a code.
Ren et al.~\cite{Ren_TSC15} propose a dynamic POR scheme for multiple storage servers
where the data file is split into data blocks and each of these data blocks is encoded
using intra-server (erasure coding) and inter-server (network coding) redundancy.
An update in a block requires changing only a few codeword symbols.
Moreover, the inter-server redundancy achieved using network coding reduces
the repair bandwidth required in case any of the servers fails.
The POR scheme by Guan et al.~\cite{Guan_ESORICS} exploits
the privately verifiable scheme of~\cite{SW_ACR} and gives
a publicly verifiable scheme with the help of indistinguishability obfuscation
($i\mathcal{O}$)~\cite{BarakIO_CR,GargIO_FOCS}.

\section{Dynamic POR Scheme with Public Verifiability}
\label{scheme}

In this section, we describe our publicly verifiable dynamic POR scheme with efficient
writes and audits.
Like the existing dynamic POR schemes~\cite{Wichs_ORAM,Stefanov_CCS}, our
construction also rely on the hierarchical structure provided by the
oblivious RAM~\cite{ORAM_JACM}. Specifically, we follow a storage
structure similar to the one proposed by Shi et al.~\cite{Stefanov_CCS}.
However, our construction is more efficient than their scheme in terms
of the cost of a write operation and the cost of an audit.
Our construction is based on \textit{collision-resistant homomorphic hashing}
technique~\cite{Krohn_SP,IncCrypto_CR} along with a digital signature scheme.
To the best of our knowledge, the homomorphic hashing technique has not been used
before in the context of POR schemes.

\subsection{Storage Structure for Data Blocks}
\label{storage_blocks}
Our scheme relies on a storage structure similar to that proposed by Shi et al.~\cite{Stefanov_CCS}.
Let the client (data owner) associate a random file-identifier \texttt{fid} of $\lambda$ bit-size to the data file she
wants to outsource to the cloud server. We assume that the data file is divided into blocks of size
$\beta$ bits, and read (and write) operations are performed on these blocks.
The value of $\beta$ is taken to be $\flr{\log\tilde{p}}$ for a large prime $\tilde{p}$. The way this prime
$\tilde{p}$ is selected is discussed in Section~\ref{buff_H}. For static data, a standard way
to guarantee retrievability of the file is to encode the file using an erasure code~\cite{SW_ACR}.
The main drawback of using erasure codes in \textit{dynamic} POR is that an update in a single block
in a codeword (say, C) is followed by updates on other $O(n)$ blocks in C,
where $n$ is the number of blocks being encoded to form C.
The underlying idea to overcome this drawback is not to update the encoded copy (C)
for every write operation (insertion, deletion or modification). Instead, it is updated (or rebuilt) only when
sufficient updates are done on the data file. Thus, the amortized cost for writes is
reduced dramatically. However, this encoded copy stores stale data between two such rebuilds.
Therefore, a hierarchical log structure is maintained which temporarily
stores values for the intermediate writes between two successive rebuilds of C.
Each level of this hierarchical log
is also encoded using an erasure code.

We adopt the storage structure and code construction mentioned above in our scheme.
However, we use collision-resistant homomorphic hashing to construct another hierarchical
storage (discussed in Section~\ref{storage_tags}) in order to reduce the cost of a write
and an audit for the client.
Our scheme involves the following three data structures in order to store the data blocks
of the client's file:
\begin{itemize}
\item an \textit{unencoded} buffer U containing all the up-to-date data blocks of the file
(that is, U is updated after every write operation is performed),

\item an \textit{encoded} buffer C which is updated after every $n$ writes
(that is, C is rebuilt afresh by encoding the latest U after every $n$ write operations), and

\item an
\textit{encoded} hierarchical (multiple levels of buffers) log structure H which accommodates
all intermediate writes between two successive rebuilds of C (H is made empty after every $n$ writes).
\end{itemize}
We note that \textit{all} of these data structures are stored on the cloud server.
The server also stores two additional data structures, $\tilde{\text{H}}$ and $\tilde{\text{C}}$
(similar to H and C, respectively), for authentication tags described in Section~\ref{storage_tags}.

\subsubsection{Structure of Buffer U}
\label{buff_U}
The buffer U contains an up-to-date copy of the data file. Reads and writes are
performed directly on the required locations of U. A Merkle hash tree is maintained
over the data blocks of U to check the authenticity of the read block (see Section~\ref{MHT} for the description
of a Merkle hash tree). The Merkle proof sent by the server along with the read block
is verified with respect to the up-to-date root-digest (say, $digMHT$) of the Merkle hash tree.
One can also use other authenticated data structures
like rank-based authenticated skip lists~\cite{Erway_TISSEC} instead of a Merkle hash tree.
Let $n$ be the number of blocks the client outsources to the cloud server initially. So the height
of the Merkle tree built on U is $\ceil{\log n}$.
Read and write operations on the buffer U are described in details in Section~\ref{operations}.

\subsubsection{Structure of Hierarchical Log H}
\label{buff_H}
A hierarchical log structure H is maintained that consists of $(k+1)$ levels
$\text{H}_0,\text{H}_1,\ldots,\text{H}_k$, where $k=\flr{\log n}$.
The log structure H stores the intermediate writes temporarily. For each $0\le l\le k$,
the $l$-th level H$_l=(X_l,Y_l)$ consists of an encoded copy of $2^l$ data blocks using
a $(2^{l+1},2^l,2^l)$-erasure code, where $X_l$ and $Y_l$ contain $2^l$ blocks each.
The original data blocks encoded in H$_l$ arrive at time $t,t+1,\ldots,t+2^l-1\mod n$,
where $t$ is a multiple of $2^l$. We describe the encoding procedure as follows.

Let $\tilde{p}$ be a large prime such that $\tilde{p}=\alpha\cdot(2n)+1$ for some $\alpha\in\N$ and
the bit-size of a block $\beta=\flr{\log\tilde{p}}$, where $\beta\gg\lambda$. Let $\tilde{g}$ denote a generator of
$\Z_{\tilde{p}}^*$. Then, $\omega=\tilde{g}^\alpha\Mod\tilde{p}$ is a $2n$-th primitive root of unity modulo $\tilde{p}$.
When a block $B$ is written to H, it is inserted in the topmost level ($l=0$) if H$_0$ is empty.
That is, $X_0$ is set to $B$. In addition, $Y_0$ is set to $B\cdot\omega^{\psi(t)}$ for the $t$-th ($\Mod n$)
write, where $\omega$ is the $2n$-th primitive root of unity modulo $\tilde{p}$.
Here, $\psi(\cdot)$ is the bit reversal function, where $\psi(t)$ is the value of the binary string
obtained by reversing the binary representation of $t$.

If the top $i$ levels $\text{H}_0,\text{H}_1,\ldots,\text{H}_{i-1}$ are already full,
a \textit{rebuild} is performed to accommodate all the blocks in these levels
as well as the new block in $\text{H}_{i}$ (and to make all the levels
up to $\text{H}_{i-1}$ empty). Shi et al.~\cite{Stefanov_CCS} employ a fast incrementally constructible code based on
Fast Fourier Transform (FFT)~\cite{FFT}\footnote{We can use any linear-time encodable and
decodable error-correcting code~\cite{Spielman} instead of the FFT-based code. However, as we
compare the performance of our scheme with that of the ``state-of-the-art'' publicly verifiable
scheme of~\cite{Stefanov_CCS} in Section~\ref{performance_ana}, we use similar code and parameters
for the ease of comparison.
}.
Fig.~\ref{fig:rebuildX} describes the algorithm
for rebuild of $X_l$ that in turn uses the algorithm \texttt{mix} shown in Fig.~\ref{fig:mixX}.
Although the algorithm deals with $X_l$, the same algorithm can be used for rebuilding $Y_l$
if we replace the $X$ arrays by corresponding $Y$ arrays and the incoming block $B$ by
$B\cdot\omega^{\psi(t)}$. We refer~\cite{Stefanov_CCS} for the form of the $(2^l\times 2^{l+1})$
generator matrix $G_l$ for the $l$-th level code. Let $\tilde{\hbox{x}}_l$ be the vector
containing $2^l$ data blocks most recently inserted in H (after applying a permutation).
Then, the output of the algorithm \texttt{mix} for H$_l$ is the same as that when
$\tilde{\hbox{x}}_l$ is multiplied by $G_l$.
Any $(2^l\times 2^l)$ submatrix of the generator matrix $G_l$ is full rank, and thus,
the code achieves the maximum distance separable (MDS) property.

As a concrete example, the rebuild of $X_3$ is demonstrated in Fig.~\ref{fig:rebuildH}.
The other part of H$_3$ (that is, $Y_3$) is rebuilt in a similar fashion.
We observe that, by using this code, the rebuild cost of H$_l$ is $O(\beta\cdot 2^l)$
(i.e., linear in the length of H$_l$) since the algorithm \texttt{mix} populates H$_l$ by combining two arrays
$\text{H}_{l-1},\text{H}'_{l-1}\in\Z_{\tilde{p}}^{2^{l-1}}$ (see Fig.~\ref{fig:mixX} and Fig.~\ref{fig:rebuildH}).
The $l$-th level is rebuilt after $2^l$ writes.
Therefore, the amortized cost for rebuilding is $O(\beta\log n)$ per write operation.
Each rebuild of the buffer C (discussed in Section~\ref{buff_C})
is followed by making all levels of H empty.

\begin{figure}[t]
\fbox{
\begin{minipage}[t]{.975\textwidth}
\renewcommand{\labelitemi}{$\bullet$}
  \begin{center}
   \textbf{Rebuild algorithm for $X_l$ to accommodate $B$ in H}
  \end{center}
  \textbf{Input}: Already full levels $X_0,X_1,\ldots,X_{l-1}$ and empty $X_l$.\newline
  \textbf{Output}: Empty levels $X_0,X_1,\ldots,X_{l-1}$ and rebuilt $X_l$.
  \begin{itemize}
   \item $X'_1\leftarrow\texttt{mix}(X_0,B,0)$
   \item \textbf{For} $i=1$ to $l-1$ \textbf{do}\newline
	 \q$X'_{i+1}\leftarrow\texttt{mix}(X_i,X'_i,i)$
   \item Make $X_0,X_1,\ldots,X_{l-1}$ empty and output $X_l=X'_l$
  \end{itemize}
\end{minipage}}
\caption{Rebuild algorithm for $X_l$.}\label{fig:rebuildX}
\end{figure}
\begin{figure}[t]
\fbox{
\begin{minipage}[t]{.975\textwidth}
\renewcommand{\labelitemi}{$\bullet$}
  \begin{center}
   \textbf{Algorithm} $\texttt{mix}(A_0,A_1,l)$
  \end{center}
  \textbf{Input}: Two arrays $A_0,A_1\in\Z_{\tilde{p}}^{2^l}$.\newline
  \textbf{Output}: Array $A$ of length $2^{l+1}$.
  \begin{itemize}
   \item Let $\omega_l=\omega^{2n/2^{l+1}}$ be the $2^{l+1}$-th primitive root of unity modulo $\tilde{p}$
   \item \textbf{For} $i=0$ to $2^l-1$ \textbf{do}
	 \begin{align}
	  A[i] & \leftarrow A_0[i]+\omega_l^iA_1[i]\mod {\tilde{p}}\\
	  A[i+2^l] & \leftarrow A_0[i]-\omega_l^iA_1[i]\mod {\tilde{p}}
	  \end{align}
   \item Output $A$
  \end{itemize}
\end{minipage}}
\caption{Algorithm \texttt{mix} for two arrays $A_0$ and $A_1$.}\label{fig:mixX}
\end{figure}

\begin{figure*}[htbp]
\centering
\fbox{\includegraphics[width=.46\textwidth]{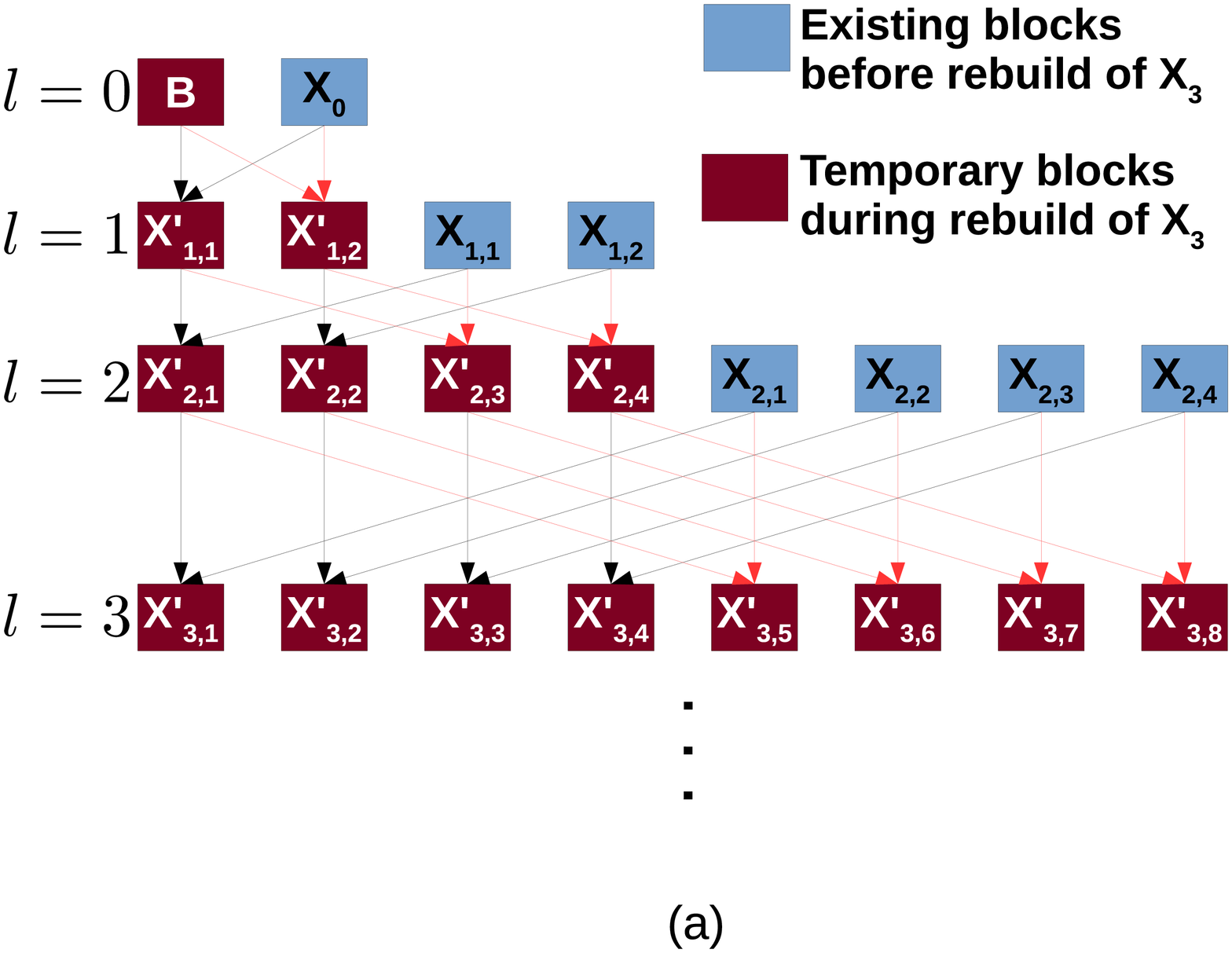}}
\qquad
\fbox{\includegraphics[width=.46\textwidth]{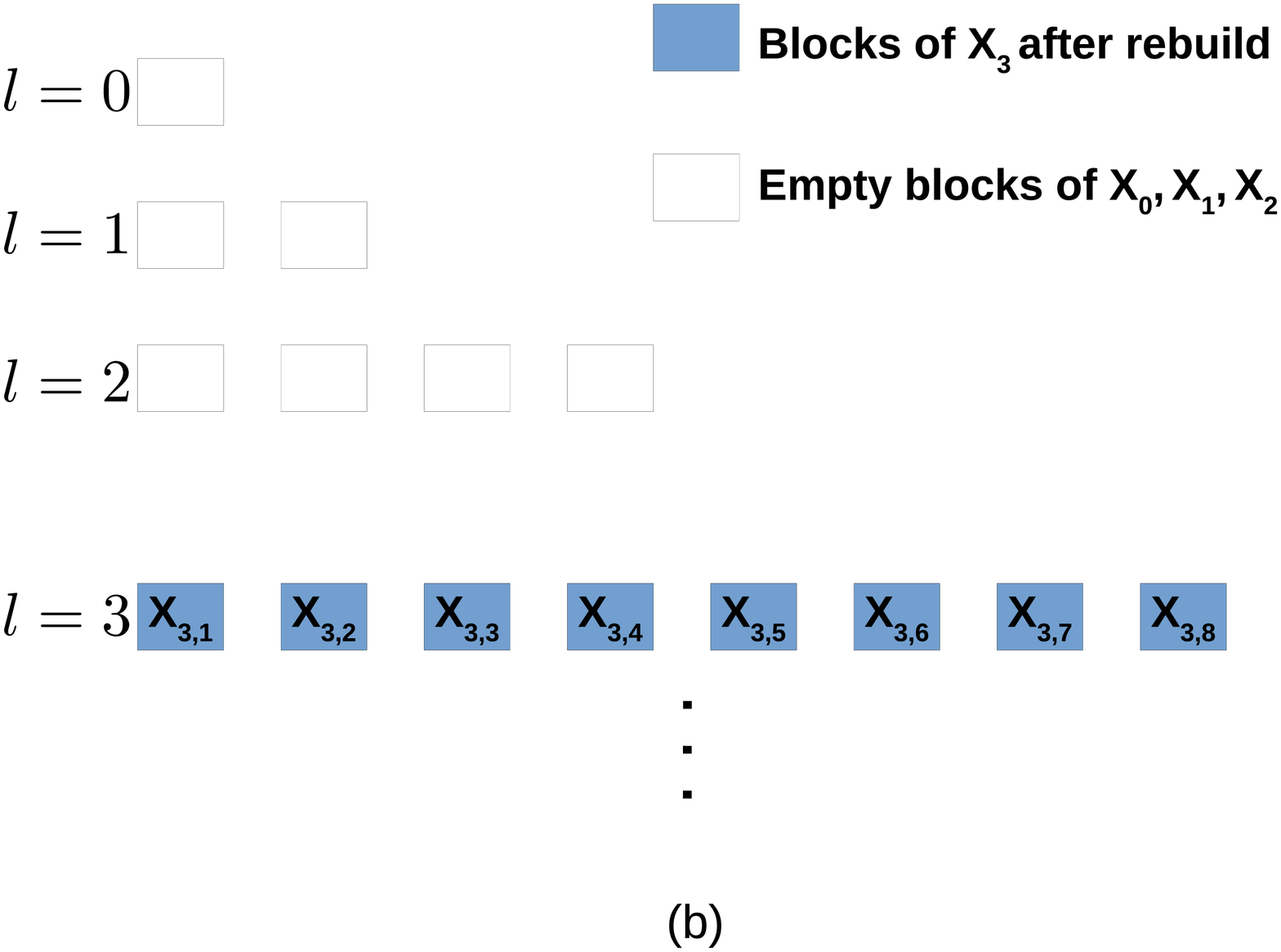}}
\caption{Rebuild process for $X_3$. (a) Initially, the levels $X_0,X_1$ and $X_2$ are already full, and $X_3$ is empty.
	The rebuild algorithm starts after a new block $B$ arrives. It forms temporary levels $X'_1,X'_2$ and $X'_3$
	using the algorithm \texttt{mix}. Linear mixing of blocks are shown using black arrows (Eqn.~1) and red arrows (Eqn.~2).
	(b) Finally, the levels $X_0,X_1$ and $X_2$ are made empty, and $X_3=X'_3$ is the newly rebuilt level.}\label{fig:rebuildH}
\end{figure*}

\subsubsection{Structure of Buffer C}
\label{buff_C}
Unlike the buffer U (and H), no read or write operations are performed directly on the buffer C.
After $n$ write operations, the buffer U is encoded using an erasure code to form a new copy of C,
and the existing copy of C is replaced by this new one. The rebuild of C can be done using the same
FFT-based code discussed in Section~\ref{buff_H} which costs $O(\beta n\log n)$ both in time and bandwidth.
As C is rebuilt after every $n$ write operations, the cost incurred per write is $O(\beta\log n)$.
We note that C contains stale data between two successive rebuilds. However, the intermediate writes
are accommodated in H with appropriate encoding. Thus, these blocks written between two successive
rebuilds of C are also retrievable at any point of time.

\subsection{Storage Structure for Authentication Tags Corresponding to Data Blocks}
\label{storage_tags}
We note that each data block in U, H and C
is of size $\beta=\flr{\log\tilde{p}}$ bits for a large prime $\tilde{p}=\alpha\cdot(2n)+1$ for some $\alpha\in\N$.
Thus, the size of a data block $\beta\gg\lambda$, where $\lambda$ is the security parameter.
For example, $\beta$ is taken to be 64 KB and $\lambda$ is taken to be 128 bits in our scheme (see Table~\ref{tab:parameters} in Section~\ref{performance_ana}).
In addition to the log structure H and the buffer C, two similar
structures $\tilde{\text{H}}$ and $\tilde{\text{C}}$ for authentication tags corresponding
to the blocks in H and C (respectively) are stored on the cloud server. Thus, the server
stores U, H, C, $\tilde{\text{H}}$ and $\tilde{\text{C}}$. The benefits of storing
$\tilde{\text{H}}$ and $\tilde{\text{C}}$ on the server are as follows.

Let us assume that the authentication tags for data blocks have the following properties.
\begin{enumerate}
 \item The size of a tag is much less than that of a block.
 \item The tags are homomorphic in the sense that, given the tags of two blocks $B_1$ and $B_2$,
       the tag on any linear combination of $B_1$ and $B_2$ can be generated.
\end{enumerate}
We note that the fundamental operation for a write (or rebuild) on H and C is to encode data blocks, that is,
to compute a linear combination of those data blocks (see Eqn.~1 and Eqn.~2 in Fig.~\ref{fig:mixX}).
Due the \textit{second} property mentioned above, while the server itself performs write operations on H and C,
the client (data owner) can perform similar operations on $\tilde{\text{H}}$ and $\tilde{\text{C}}$.
On the other hand, due to the \textit{first} property, the bandwidth required between the client and
the server for a write operation decreases significantly as the client now has to download much
smaller tags instead of large data blocks. The cost of storage is less nowadays, whereas bandwidth is
more expensive and often limited. Therefore, it is reasonable if we trade the storage off to reduce
the communication bandwidth between the client and the server required for a write (or rebuild).

Indeed, the authentication tags (described in the following section) we use in our dynamic POR scheme satisfy the following properties.
\begin{itemize}
 \item The size of a tag is $O(\lambda)$ and is independent of the size of a data block
       $\beta$, where $\lambda\ll\beta$.
 \item The homomorphic property is achieved by using a collision-resistant homomorphic hash function.
\end{itemize}
Apart from efficient write operations, the cost of an audit in our \textit{publicly verifiable dynamic} POR scheme is comparable
to that in the privately verifiable scheme of~\cite{Stefanov_CCS}, and it is \textit{much less} than that
in the publicly verifiable scheme discussed in the same work.

\subsubsection{Generation of Authentication Tags}
\label{tag_gen}

\smallskip
\noindent
\textbf{Setup}\q
For the data file identified by \texttt{fid}, the client runs an algorithm Setup($1^\lambda$)
to set parameters for generating authentication tags.
The algorithm Setup selects two random primes $p$ and $q$ such that $|q|=\lambda_q=2\lambda+1$,
$|p|=\lambda_p=O(\lambda)$ and $q|(p-1)$. Now, it divides each block $B$ of the data file into segments
of size $(\lambda_q-1)$ bits each. This ensures that each segment is less than $q$ and can
therefore be represented as an element of $\Z_q$. Thus, $m=\ceil{\beta/(\lambda_q-1)}$
is the number of segments in a block, where a block is $\beta=\flr{\log\tilde{p}}$ bits long.
In this setting, each block $B$ can be represented as a vector $[b_{1},b_{2},\ldots,b_{m}]\in \Z_q^m$.

Let $G_q$ be a subgroup of $\Z_p^*$ having order $q$. That is, $G_q$ consists of the order $q$
elements in $\Z_p$. Then, $m$ random elements $g_1,g_2,\ldots,g_m\xleftarrow{R}G_q$ are selected.
Let $\mathcal{S}=(\text{KeyGen}, \text{Sign}, \text{Verify})$ be a digital signature scheme (see Section~\ref{dig_sig})
where the algorithm Sign takes messages in $\{0,1\}^*$ as input and outputs signatures
of size $O(\lambda)$ bits each. Let the pair of signing and verification keys for
$\mathcal{S}$ be $(ssk,psk)$.
The Setup algorithm outputs the primes $p$ and $q$, the secret key $SK=ssk$,
the public parameters $PK=(g_1,g_2,\ldots,g_m,\texttt{fid},psk)$, and the descriptions of $G_q$ and $\mathcal{S}$.

\medskip
\noindent
\textbf{Format of a Tag}\q
The client computes the homomorphic hash~\cite{Krohn_SP,IncCrypto_CR} on a block $B=[b_{1},b_{2},\ldots,b_{m}]\in \Z_q^m$ as
\begin{align}\label{eqn:tag}
 h(B)=\prod\limits_{i=1}^{m}{g_i}^{b_i} \Mod p.
\end{align}
Using the signature scheme $\mathcal{S}$, the client generates
the final authentication tag for the block $B$ as
\begin{align*}
 \tilde{h}(B)=(h(B),\text{Sign}_{ssk}(h(B),\texttt{fid},\texttt{addr},t)),
\end{align*}
where \texttt{addr} is the physical address $B$ is written to (at time $t$)
and \texttt{fid} is the file-identifier of the data file the block $B$ belongs to.

\medskip
\noindent
\textbf{Collision-resistance and Homomorphic Properties}\q
As shown in~\cite{IncCrypto_CR,Krohn_SP}, given that the discrete logarithm assumption (see Section~\ref{disLog}) holds in $G_q$,
it is computationally hard to find two blocks $B_1$ and $B_2$ such that $B_1\not=B_2$ and $h(B_1)=h(B_2)$
(\textit{collision-resistance} property).

On the other hand, given $B_1=[b_{11},b_{12},\ldots,b_{1m}]\in \Z_q^m$ and
$B_2=[b_{21},b_{22},\ldots,b_{2m}]\in \Z_q^m$, any linear combination of $B_1$ and $B_2$ can be written as
$B=\alpha_1 B_1+\alpha_2 B_2=[\alpha_1 b_{11}+\alpha_2 b_{21},\alpha_1 b_{12}+\alpha_2 b_{22},\ldots,\alpha_1 b_{1m}+\alpha_2 b_{2m}]\in \Z_q^m$.
Therefore, $h(B)$ can be computed as
$\prod_{i=1}^{m}{g_i}^{\alpha_1 b_{1i}+\alpha_2 b_{2i}} \Mod p=h(B_1)^{\alpha_1}\cdot h(B_2)^{\alpha_2}$
(\textit{homomorphic} property).

\medskip
\noindent
\textbf{Size of a Tag}\q
The size of an authentication tag $\tilde{h}(B)$ is the sum of $|h(B)|$ (which is $\lambda_p=O(\lambda)$ bits)
and the size of a signature in $\mathcal{S}$.
If we use the standard ECDSA signature scheme~\cite{ECDSA}
as $\mathcal{S}$, then a signature is $4\lambda$ bits long\footnote{To reduce the size of
a tag, we can use short signatures of size $2\lambda$ bits~\cite{BLS_JOC}. However, the verification of a signature
is more expensive due to computation of bilinear pairings~\cite{Galbraith_DAM}.}.
Thus, $|\tilde{h}(B)|$ is also $O(\lambda)$ bits.
For the values of different parameters considered in our scheme
(see Table~\ref{tab:parameters} in Section~\ref{performance_ana}),
the size of a tag is only 192 bytes which is very small compared to the size of a block (64 KB).

\medskip
\noindent
\textbf{Improvement in Cost of Tag Computation}\q
To compute the homomorphic hash $h(B)$ on a block $B$ using Eqn.~\ref{eqn:tag}, it requires $m$ exponentiations and $(m-1)$
multiplications modulo $p$. We can reduce this computational complexity in the following way
at the cost of the client storing $m$ elements of $\Z_q^*$ which is essentially equivalent to just one block.
The client chooses $g\xleftarrow{R}G_q$ and $\gamma_1,\gamma_2,\ldots,\gamma_m\xleftarrow{R}\Z_q^*$.
The client sets $g_i=g^{\gamma_i} \Mod p$ for each $i\in[1,m]$. The client includes the vector
$\Gamma=[\gamma_1,\gamma_2,\ldots,\gamma_m]$ and $g$ in her secret key $SK$ and makes $g_1,g_2,\ldots,g_m$ public as before.
Now, the homomorphic hash $h(B)$ on a block $B$ is computed as 
\begin{align}\label{eqn:Imp}
 h(B)& = \prod\limits_{i=1}^{m}{g}^{\gamma_i b_i} \Mod p\notag\\
     & = g^{\sum_{i=1}^{m}{\gamma_i b_i}} \Mod p
\end{align}
which requires only \textit{one} exponentiation modulo $p$
along with $m$ multiplications and $(m-1)$ additions modulo $q$. This is a huge improvement
in the cost for computing an authentication tag. On the other hand, the storage overhead at the client's side
is $|\Gamma|$ which is same as the size of a single block $B$. Considering the fact that the client outsources
millions of blocks to the cloud server, this amount of client storage is reasonable for all practical
purposes.

\subsubsection{Storage Structure for $\tilde{\text{H}}$ and $\tilde{\text{C}}$}
\label{buff_tags}
The storage structures for $\tilde{\text{H}}$ and $\tilde{\text{C}}$ are exactly the same
as those for H and C, respectively, except that the authentication tags
(instead of data blocks) are stored in $\tilde{\text{H}}$ and $\tilde{\text{C}}$
(see Section~\ref{storage_blocks} for structures of H and C).

\subsection{Operations}
\label{operations}
There are three types of operations involved in a dynamic POR scheme. The client can read, write and audit
her data stored on the cloud server. The read and write operations are \textit{authenticated} in that the client
can verify the authenticity of these operations. We note that though the client herself performs reads and writes on her data,
she can delegate the auditing task to a third party auditor (TPA) for a publicly verifiable scheme. As our scheme
is \textit{publicly verifiable}, we use the term \textit{verifier} to denote an auditor who can be a TPA or the client herself.
Fig.~\ref{fig:flow} gives an overview of the communication flow between the client and the server during these
operations.

\begin{figure*}[t]
\centering
\fbox{\includegraphics[width=.46\textwidth]{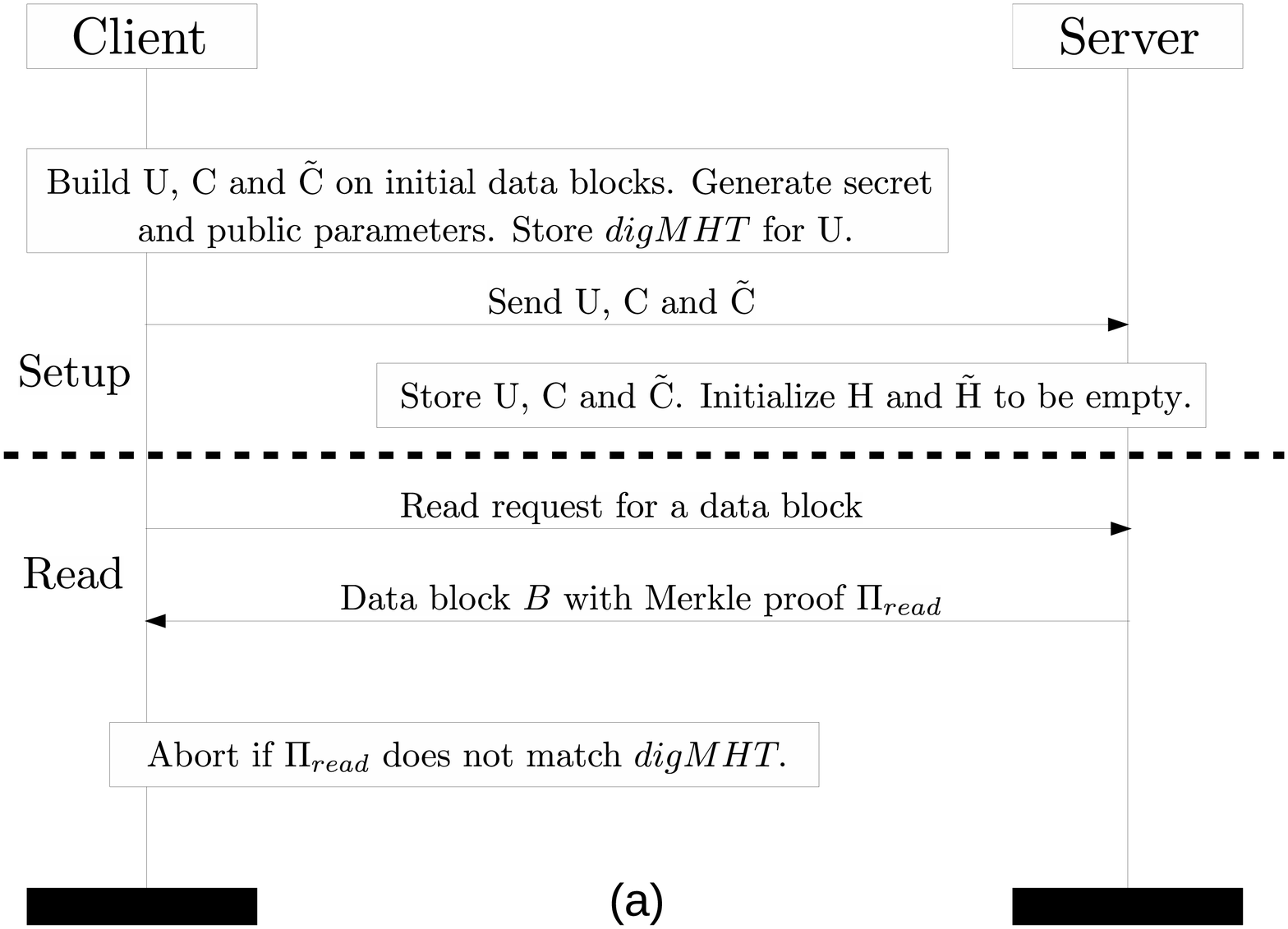}}
\qquad
\fbox{\includegraphics[width=.46\textwidth]{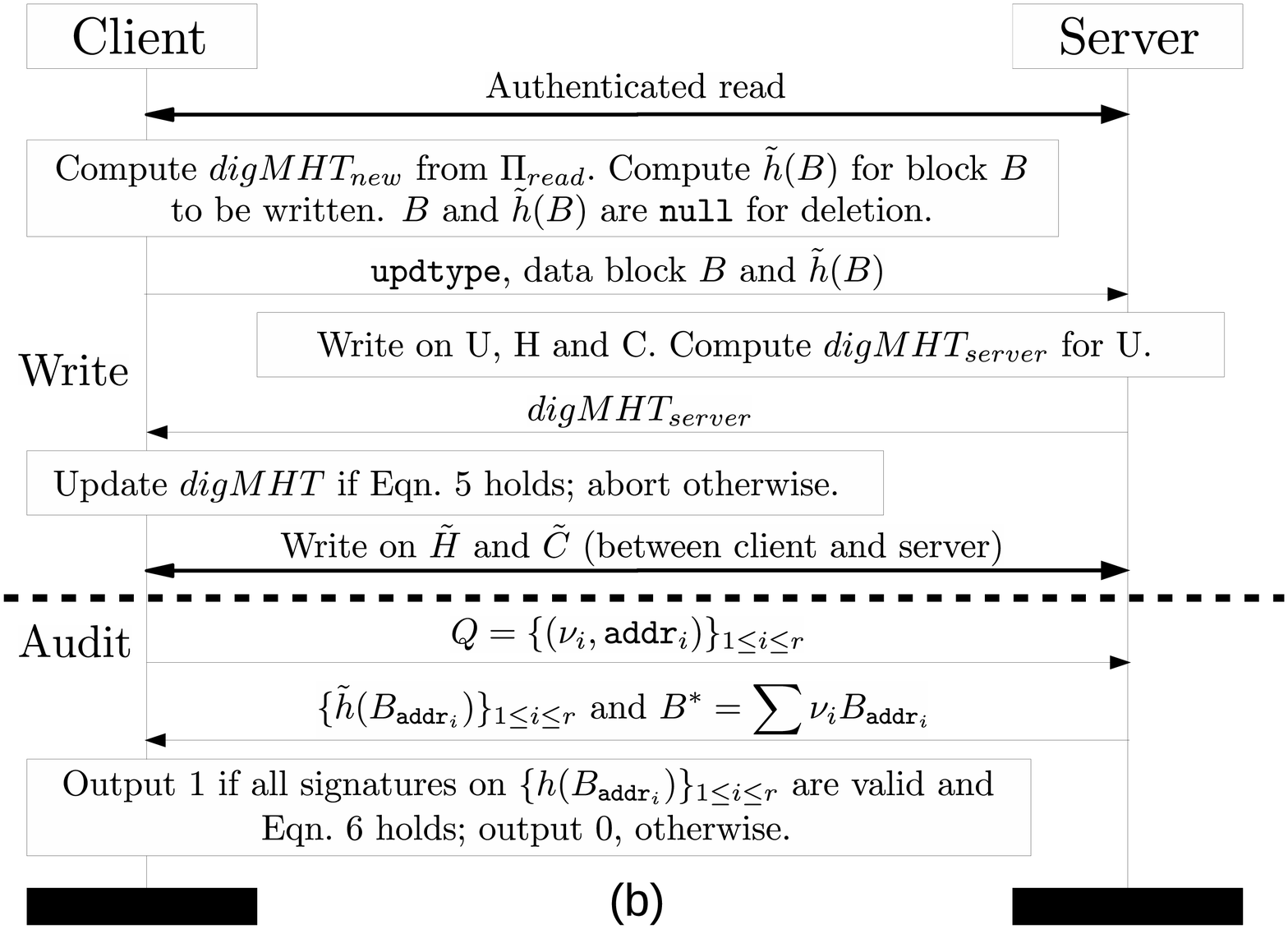}}
\caption{Communication flow between the client and the server for different operations described in Section~\ref{operations}.
	In the setup phase, the client sets parameters for the scheme and outsources the preprocessed file to the server.
	Initially, the client uploads U, C and $\tilde{\text{C}}$.
	The server stores them along with H and $\tilde{\text{H}}$ that are initialized to be empty.
	Then, the client can perform reads, writes and audits on her data in an \textit{interleaved} fashion.
	We note that, during a write, the server itself rebuilds H and C (if necessary). On the other hand, the client
	rebuilds $\tilde{\text{H}}$ and $\tilde{\text{C}}$ by downloading some of the authentication tags from them,
	computing the tag on the new block (using homomorphic property of tags) and sending the new tag
	to the server. As our scheme is publicly verifiable, audits can be performed by a third party
	auditor (TPA) as well.}\label{fig:flow}
\end{figure*}

\subsubsection{Read}
\label{subsec:read}
Reads are performed directly from the unencoded buffer U. The authenticity and freshness of the block
read can be guaranteed by using a Merkle hash tree~\cite{Merkle_CR} (or a similar data structure like
rank-based authenticated skip list~\cite{Erway_TISSEC}) over the blocks of U.
That is, the blocks of U constitute the leaf-nodes of the Merkle hash tree
(see Section~\ref{MHT} for a brief description of a Merkle hash tree).
The server sends the Merkle proof $\Pi_{read}$ containing the requested block and the labels of the nodes along
the associated path of the tree to the client. The client maintains the up-to-date value of the root-digest
of the Merkle hash tree $digMHT$ and verifies the proof $\Pi_{read}$ with respect to this root-digest.
We note that the size of the root-digest of the Merkle hash tree is $O(\lambda)$ bits.

\subsubsection{Write}
\label{subsec:write}
Let \texttt{updtype} denote the type of a write operation which can be insertion of a new block after the $i$-th block,
deletion of the $i$-th block or modification of the $i$-th block. A write operation affects the buffers in the following way.

\begin{itemize}
\item \textbf{Write on U}\q As the buffer U is unencoded, a write operation on U can be performed
in a similar way as done on the data blocks in a dynamic provable data possession (PDP) scheme~\cite{Erway_TISSEC}.
We briefly describe the procedure as follows.

Let $digMHT$ be the root-digest of the current Merkle hash tree which is stored at the client's side.
The client performs an \textit{authenticated} read on the $i$-th data block of U (as described above).
If the associated Merkle proof $\Pi_{read}$ does not match with $digMHT$, the client aborts.
Otherwise, she computes, from $\Pi_{read}$, the value that would be the new root-digest (say, $digMHT_{new}$)
if the write operation is correctly performed on U stored at the server.
The client stores $digMHT_{new}$ temporarily at her end and asks the server to perform the write on U.
The server performs this write operation on U, computes the root-digest $digMHT_{server}$ of the Merkle hash tree
and sends $digMHT_{server}$ to the client. The client verifies whether
\begin{align}\label{eqn:MHT}
digMHT_{new}\stackrel{?}{=}digMHT_{server}.
\end{align}
If they are not equal, the client aborts. Otherwise, the client sets $digMHT=digMHT_{new}$ at her end.

\item \textbf{Write on H and $\tilde{\text{H}}$}\q We assume that deletion of a block in U
corresponds to insertion of a block (with \texttt{null} content) in the hierarchical log H.
Therefore, for a write of any \texttt{updtype} (i.e., insertion, deletion or modification),
only insertions take place in H. The way a (possibly encoded) block $B$ is inserted in H
is discussed in details in Section~\ref{buff_H}. The cloud server itself performs this operation on H.

An insertion in $\tilde{\text{H}}$ is performed by the client herself as this procedure requires
the knowledge of secret information held by the client. The client computes the authentication tag
on the (possibly encoded) block and insert it in $\tilde{\text{H}}$.
The underlying basic operation of the rebuild phase of H is to compute a linear combination
$B$ (e.g., $\alpha_1 B_1+\alpha_2 B_2)$ of two blocks $B_1$ and $B_2$ (see Eqn.~1 and~2 in Fig.~\ref{fig:mixX}).
Similarly, the corresponding operation for the rebuild of $\tilde{\text{H}}$ is to compute
$\tilde{h}(B)$ given $\tilde{h}(B_1)$ and $\tilde{h}(B_2)$. For $i=1,2$, the client first
downloads $\tilde{h}(B_i)$ and verifies the signature on $h(B_i)$ by checking whether
\begin{align*}
 \text{Verify}_{psk}((h(B_i),\texttt{fid},\texttt{addr}_i,t_i),\tilde{h}(B_i))\stackrel{?}=\texttt{accept},
\end{align*}
where $psk$ is the verification key for the signature scheme $\mathcal{S}$.
We note that \texttt{addr}$_1$ (or \texttt{addr}$_2$) is the physical address of the block $B_1$ (or $B_2$)
written at time $t_1$ (or $t_2$), and \texttt{fid} is the file-identifier of the data file the block $B$ belongs to.
For any block in H and C, the time when the block was written most recently can be easily computed
from the current time itself. If the verification passes,
the client computes $h(B)=h(B_1)^{\alpha_1}\cdot h(B_2)^{\alpha_2}$ and $\tilde{h}(B)$ subsequently.
This requires two exponentiations and one multiplication modulo $p$ along with one Sign and two Verify operations.

\item \textbf{Write on C and $\tilde{\text{C}}$}\q As mentioned in Section~\ref{buff_C}, C
($\tilde{\text{C}}$ for authentication tags) is rebuilt after every $n$ writes.
The server performs a rebuild on C, whereas a rebuild on $\tilde{\text{C}}$ is performed by the client.
Basic operations involved in rebuilds of C and $\tilde{\text{C}}$ are the same
as those for rebuilds of H and $\tilde{\text{H}}$, respectively, and thus are omitted here.

\end{itemize}

\subsubsection{Audit}
\label{subsubsec:audit}
In the \textit{challenge} phase, the verifier chooses $r=O(\lambda\log n)$ random locations $\{\texttt{addr}_i\}_{1\le i\le r}$
from all levels (where $O(\lambda)$ random locations are selected from each level) of H and C.
Then, she sends a challenge set $Q=\{(\nu_i,\texttt{addr}_i)\}_{1\le i\le r}$
to the cloud server, where $\nu_1,\nu_2,\ldots,\nu_r\xleftarrow{R}\Z_q^*$ are random coefficients.
In the \textit{response} phase, the server sends to the verifier a proof containing $B^*=\sum_{1\le i\le r}{\nu_i}B_{\texttt{addr}_i}$
and $\{\tilde{h}(B_{\texttt{addr}_i})\}_{1\le i\le r}$. Upon receiving the proof from
the server, the verifier verifies each of the signatures on $\{h(B_{\texttt{addr}_i})\}_{1\le i\le r}$.
Then, she computes $h^*=\prod_{1\le i\le r}{h(B_{\texttt{addr}_i})}^{\nu_i}$ and $h(B^*)$ using Eqn.~\ref{eqn:tag}.
Finally, the verifier checks whether
\begin{align}\label{eqn:verAudit}
h(B^*)\stackrel{?}=h^*
\end{align}
and outputs 0 if any of the verifications fails; she outputs 1, otherwise.

\section{Security}
\label{security}

We define security of a dynamic POR scheme~\cite{Stefanov_CCS,Wichs_ORAM} and show that our scheme described in Section~\ref{scheme}
is secure according to this definition. We also show that the server cannot pass an audit without
storing \textit{all} data blocks properly, except with some probability negligible in $\lambda$.

\subsection{Overview of Security of a Dynamic POR Scheme}
\label{security_overview}
A dynamic POR scheme must satisfy the following properties~\cite{Stefanov_CCS}.
The formal security definition is given in Section~\ref{security_model}.
\begin{enumerate}

\item \textbf{Authenticity and Freshness}\q The authenticity property requires that the cloud server cannot
produce valid proofs during 
audits without storing
the corresponding blocks and their respective authentication information untampered, except with
a probability negligible in $\lambda$.

For dynamic data,
the client can modify an existing data block. However, a malicious
cloud server may discard this change and keep an old copy of the block. As the old copy
of the block and its corresponding tag constitute a valid pair, the client has no way to detect
if the cloud server is storing the \textit{fresh} (latest) copy. Thus, the client must be convinced
that the server has stored the up-to-date blocks.\smallskip

\item \textbf{Retrievability}\q Retrievability of data requires that, given a probabilistic
polynomial-time adversary $\mathcal{A}$ that can respond correctly to a challenge $Q$
with some non-negligible probability,
there exists a polynomial-time extractor algorithm $\mathcal{E}$ that can extract \textit{all}
data blocks of the file (except with negligible probability) by challenging $\mathcal{A}$ for a polynomial
(in $\lambda$) number of times and verifying the responses sent by $\mathcal{A}$.
The algorithm $\mathcal{E}$ has a black-box rewinding access to $\mathcal{A}$.
  Authenticity and freshness of data restrict the adversary $\mathcal{A}$
  to produce valid responses (without storing the data in an authentic and up-to-date fashion) during these interactions
  only with some probability negligible in the security parameter $\lambda$.

\end{enumerate}

\subsection{Security Model}
\label{security_model}
We first describe the following security game between the challenger (acting as the client)
and the adversary (acting as the cloud server).

\begin{itemize}

\item The adversary selects a file $F$ associated with a file-identifier \texttt{fid} to store.
The challenger processes the file to form another file $F'$ and returns $F'$ to the adversary.
The challenger stores only some metadata for verification purpose.

\item The adversary adaptively chooses a sequence of operations defined by
$\{\texttt{op}_i\}_{1\le i\le q_1}$ ($q_1$ is polynomial in
the security parameter $\lambda$), where $\texttt{op}_i$ is a read, a write or an audit.
The challenger executes these operations on the file stored by the adversary.
For each operation, the challenger verifies the response sent by the adversary
and updates the metadata at her end only if the response passes the verification.

\item Let $F^*$ be the final state of the file after $q_1$ operations.
The challenger has the latest metadata for the file $F^*$. Now, she executes an audit
protocol with the adversary. The challenger sends a random challenge set $Q$ to the adversary,
and the adversary returns a cryptographic proof to the challenger.
The adversary wins the game if it passes the verification.

\end{itemize}

\begin{definition}[\textbf{Security of a Dynamic POR Scheme}]\label{def:security_dpor}
A dynamic POR scheme is secure if, given any probabilistic polynomial-time
adversary $\mathcal{A}$ who can win the security game mentioned above with some non-negligible
probability, there exists a polynomial-time extractor algorithm $\mathcal{E}$ that can extract
all data blocks of the file by interacting (via challenge-response) with $\mathcal{A}$
polynomially many times.
\end{definition}

\subsection{Security Analysis of Our Scheme}
\label{security_analysis}
We state and prove the following theorem in order to analyze the security of our dynamic POR scheme.

\begin{theorem}\label{theorem_DPOR}
Given that the discrete logarithm assumption holds in $G_q$ and the underlying digital signature scheme is secure,
the dynamic POR scheme described in Section~\ref{scheme} is secure according to Definition~\ref{def:security_dpor}.
\end{theorem}

\begin{proofTheorem}
We use the following claim in order to prove Theorem~\ref{theorem_DPOR}.

\begin{claim}\label{claim_authenticity}
Given that the discrete logarithm assumption holds in $G_q$ and the underlying digital signature scheme is secure,
authenticity and freshness of the challenged blocks in H and C are guaranteed.
\end{claim}

\begin{proof}
We prove the above claim for the log structure H. The proof for C follows in a similar way.
In our scheme, every block $B$ (of the file identified by \texttt{fid}) in H corresponds to an authentication tag
$\tilde{h}(B)=(h(B),\text{Sign}_{ssk}(h(B),\texttt{fid},\texttt{addr},t))$ present in $\tilde{\text{H}}$,
where the signing algorithm Sign uses the secret key $ssk$ of the client and $t$ is the last write-time of the block $B$.
Let $B$ be the correct block that was actually written by the client to the address \texttt{addr} at time $t$. 
Suppose this block in \texttt{addr} is challenged during an audit.
We note that the last write-time $t$ of the block is computable from \texttt{addr} and the current time.
So, the values of \texttt{fid}, \texttt{addr} and $t$ are known to the challenger.
Therefore, in order to break the authenticity of the scheme,
the PPT adversary $\mathcal{A}$ has to find a block $B'\not=B$ and its tag $\tilde{h}(B')$ such that one of the following conditions holds:
\begin{itemize}
 \item Case I: $h(B')\not=h(B)$ and $\tilde{h}(B')=(h(B'),\text{Sign}_{ssk}(h(B'),\texttt{fid},\texttt{addr},t))$,
 \item Case II: $h(B')=h(B)$.
\end{itemize}

\paragraph{Case I}\quad
We show that, if the adversary $\mathcal{A}$ can find a block $B'\not=B$ and its authentication tag $\tilde{h}(B')$ such that 
$h(B')\not=h(B)$ and $\tilde{h}(B')=(h(B'),\text{Sign}_{ssk}(h(B'),\texttt{fid},\texttt{addr},t))$, 
then it can break the security of the underlying signature scheme
(the security of a digital signature scheme is discussed in Section~\ref{dig_sig}).

Let the adversary $\mathcal{A}$ be provided with a set of 
polynomially many 
authentication tags $\{\tilde{h}(B_i)=(h(B_i),\text{Sign}_{ssk}(h(B_i),\texttt{fid},\texttt{addr}_i,t_i))\}_{i\in I}$
for $\{(B_i,\texttt{addr}_i,t_i)\}_{i\in I}$
of $\mathcal{A}$'s choice
(where $I=[1,k]$ for some $k$ polynomial in $\lambda$).
Let us assume that the adversary $\mathcal{A}$ is able to find another block $B'$ and its tag 
$\tilde{h}(B')=(h(B'),\text{Sign}_{ssk}(h(B'),\texttt{fid},\texttt{addr}_j,t_j))$, 
such that $j\in I$, $B'\not=B_j$ and $h(B')\not=h(B_j)$.
Then, we can construct another probabilistic polynomial-time (PPT)
algorithm $\mathcal{B}^{\mathcal{O}_{ssk}(\cdot)}$ that, given the public key $psk$ and an access to the signing oracle $\mathcal{O}_{ssk}(\cdot)$, 
executes $\mathcal{A}$ as a subroutine. 
Initially, $\mathcal{B}$ provides the 
public parameters $(g_1,g_2,\ldots,g_m,\texttt{fid},psk)$ and
the description of $G_q$ to $\mathcal{A}$.
With the help of $\mathcal{O}_{ssk}(\cdot)$, $\mathcal{B}$
responds to $\mathcal{A}$'s queries with $\{\tilde{h}(B_i)=(h(B_i),\text{Sign}_{ssk}(h(B_i),\texttt{fid},\texttt{addr}_i,t_i))\}_{i\in I}$.
Now, if $\mathcal{A}$ finds another block $B'$ and its tag 
$\tilde{h}(B')=(h(B'),\text{Sign}_{ssk}(h(B'),\texttt{fid},\texttt{addr}_j,t_j))$ 
as described above with probability $\epsilon_{\mathcal{A}}$ in (polynomial) time $t'_{\mathcal{A}}$, 
then $\mathcal{B}$ also finds a forged signature $\text{Sign}_{ssk}(h(B'),\texttt{fid},\texttt{addr}_j,t_j)$
(that was not queried to the signing oracle before) 
with probability $\epsilon_{\mathcal{B}}=\epsilon_{\mathcal{A}}$ in time $t'_{\mathcal{B}}\approx t'_{\mathcal{A}}$.

\paragraph{Case II}\quad
We show that, if the adversary $\mathcal{A}$ can find a block $B'\not=B$ and 
its authentication tag $\tilde{h}(B')=(h(B'),\text{Sign}_{ssk}(h(B'),\texttt{fid},\texttt{addr},t))$ 
such that $h(B')=h(B)$, 
then it can solve the discrete logarithm problem over $G_q$
(we refer to~\cite{IncCrypto_CR,Krohn_SP} for the detailed proof showing that the collision-resistance property holds for $h$).

The idea of the proof is as follows. Let us assume that the adversary $\mathcal{A}$,
given the description of the multiplicative group $G_q=\langle g \rangle$ and
$m$ random elements $g_1,g_2,\ldots,g_m$ of $G_q$,
is able to find two blocks
$B,B'\in \Z_q^m$
such that $B\not=B'$ and $h(B)=h(B')$. Then, we can construct another probabilistic polynomial-time (PPT) 
algorithm $\mathcal{B}$ that, given the description of $G_q$ 
and $y\in G_q$, executes $\mathcal{A}$ as a subroutine to find a collision and uses this collision to compute $x\in\Z_q$ such that $y=g^x$.
In order to do that, $\mathcal{B}$ selects $z_1,z_2,\ldots,z_m\xleftarrow{R}\{0,1\}$
and $u_1,u_2,\ldots,u_m\xleftarrow{R}\Z_q$. For each $i\in [1,m]$, $\mathcal{B}$ sets
$g_i=g^{u_i}$ if $z_i=0$; it sets $g_i=y^{u_i}$  if $z_i=1$.
Then, $\mathcal{B}$ provides $\mathcal{A}$ with the description of $G_q=\langle g \rangle$ and
the elements $g_1,g_2,\ldots,g_m\in G_q$ computed in the previous step.
Now, suppose $\mathcal{A}$ finds two blocks 
$B=[b_1,b_2,\ldots,b_m]\in \Z_q^m$ and $B'=[b'_1,b'_2,\ldots,b'_m]\in \Z_q^m$ 
with probability $\epsilon_{\mathcal{A}}$ in (polynomial) time $t'_{\mathcal{A}}$,
such that $B\not=B'$ and $h(B)=h(B')$.
Then, $\mathcal{B}$ sets $a=\sum_{z_i=1}u_i(b_i-b'_i) \Mod q$ and computes $a'=a^{-1}\Mod q$ 
($a$ is non-zero with probability at least $\frac{1}{2}$).
Since $h(B)=h(B')$, we have
\begin{align*}
	     & \prod_{i=1}^{m}{g_i^{b_i}} = \prod_{i=1}^{m}{g_i^{b'_i}}\\
    \implies & \prod_{z_i=1}{y^{u_i(b_i-b'_i)}} = \prod_{z_i=0}{g^{u_i(b'_i-b_i)}}\\
    \implies & y^a = \prod_{z_i=0}{g^{u_i(b'_i-b_i)}}\\
    \implies & y^{aa'} = \prod_{z_i=0}{g^{a'u_i(b'_i-b_i)}}\\
    \implies & y = g^x,
   \end{align*}
where $x=\sum_{z_i=0}{a'u_i(b'_i-b_i)} \Mod q$ is the discrete logarithm of $y$ in $G_q$.
Thus, the algorithm $\mathcal{B}$ solves the discrete logarithm problem over $G_q$ 
with probability $\epsilon_{\mathcal{B}}\ge\frac{\epsilon_{\mathcal{A}}}{2}$ 
in (polynomial) time $t'_{\mathcal{B}}=t'_{\mathcal{A}}+O(m{\lambda}^3)$. 
The overhead term $O(m{\lambda}^3)$ is attributed to some arithmetic operations (including $m$ exponentiation operations)
that $\mathcal{B}$ has to perform.

Given an address \texttt{addr}, let $B$ be the latest block that was actually written by the client to \texttt{addr} at time $t$.
Let the challenger challenge the block in \texttt{addr} during an audit.
In order to retain an older block $B'\not=B$ (written to the same address \texttt{addr} at time $t'<t$) and still pass the audit,
the adversary $\mathcal{A}$ has to produce its authentication tag for time $t$ (we note that the tag for $B'$ for time $t'$
is available to $\mathcal{A}$) such that one of the conditions mentioned above (Case I and Case II) holds.
As we have seen earlier, it is computationally hard to find such a block $B'$, except with a probability negligible in $\lambda$.
Thus, the adversary must store each of the challenged blocks with its latest content to pass the audit.
\end{proof}

We define a polynomial-time extractor algorithm $\mathcal{E}$ that can extract
all blocks from each of the levels of H and C (except with negligible probability) by interacting with an adversary $\mathcal{A}$
that wins the security game described in Section~\ref{security_model} with some non-negligible probability.
As our dynamic POR scheme satisfies the \textit{authenticity}
and \textit{freshness} properties mentioned above,
the adversary $\mathcal{A}$ cannot produce a valid proof $(B^*=\sum_{1\le i\le r}{\nu_i}B_{\texttt{addr}_i},\{\tilde{h}(B_{\texttt{addr}_i})\}_{1\le i\le r})$
for a given challenge set $Q=\{(\nu_i,\texttt{addr}_i)\}_{1\le i\le r}$ without storing
the challenged blocks and their corresponding tags properly,
except with some negligible probability (see Section~\ref{subsubsec:audit} and Claim~\ref{claim_authenticity}).
This means that if the verifier outputs 1 during the extraction phase,
$B^*$ in the proof is indeed the linear combination of the untampered blocks $\{B_{\texttt{addr}_i}\}_{1\le i\le r}$
using coefficients $\{\nu_i\}_{1\le i\le r}$.

Suppose that the extractor $\mathcal{E}$ wants to extract
$r$ blocks indexed by $J$. It challenges $\mathcal{A}$ with a challenge set $Q=\{(\nu_i,\texttt{addr}_i)\}_{i\in J}$.
If the proof is valid (that is, the verifier outputs 1), $\mathcal{E}$ initializes a matrix
$M_\mathcal{E}$ as $[\nu_{1i}]_{i\in J}$, where $\nu_{1i}=\nu_{i}$ for each $i\in J$.
The extractor challenges $\mathcal{A}$ for the same $J$ but with different random coefficients.
If the verifier outputs 1 and the vector of coefficients is linearly independent to
the existing rows of $M_\mathcal{E}$, then $\mathcal{E}$ appends this vector to $M_\mathcal{E}$ as a row.
The extractor $\mathcal{E}$ runs this procedure until the matrix $M_\mathcal{E}$ has $r$ linearly independent rows.
So, the final form of the full-rank matrix $M_\mathcal{E}$ is $[\nu_{ji}]_{j\in[1,r], i\in J}$.
Consequently, the challenged blocks can be extracted using Gaussian elimination.

Following the way mentioned above, the extractor algorithm $\mathcal{E}$ can interact with $\mathcal{A}$
(polynomially many times) in order to extract $\rho$-fraction of blocks
(for some $\rho$) for each level of H and C by setting the index set $J$ appropriately.
Use of a $\rho$-rate erasure code ensures retrievability of all blocks
of C (i.e., all the encoded blocks of U up to the last rebuild of C) and H (i.e., all the encoded blocks
of U written after the last rebuild of C).
For each $l$-th level of H (or C), the FFT-based code used in our scheme is a $(2^{l+1},2^l,2^l)$-erasure code;
thus, $\rho=\frac{1}{2}$.

This completes the proof of Theorem~\ref{theorem_DPOR}.
\end{proofTheorem}

\subsection{Probabilistic Guarantees}
As we mention in Section~\ref{operations}, each of the levels of H and the buffer C is audited with $O(\lambda)$ random locations.
Due to the use of a $(2^{l+1},2^l,2^l)$-erasure code for each level $0\le l\le \flr{\log n}$,
the server has to actually delete half of the blocks in a level in order to delete a single block in that level.
Thus, if the server corrupts half of the blocks in any level, then
it passes an audit with probability $p_{cheat}=(1-\frac{1}{2})^{O(\lambda)}=2^{-O(\lambda)}$
that is negligible in $\lambda$.

\section{Performance Analysis}
\label{performance_ana}

We analyze the performance of the following types of operations (described in Section~\ref{operations})
involved in our publicly verifiable dynamic POR scheme.

\begin{itemize}
\item \textbf{Read}\q
For an authenticated read on the data block $B$ present in U, the server sends the corresponding Merkle proof $\Pi_{read}$
which consists of the block $B$, the data block in the sibling leaf-node of $B$ and the hash values along the associated path of the Merkle
hash tree (see Section~\ref{MHT}).
Thus, a read operation takes $2\beta+O(\lambda\log n)$ communication bandwidth between the client and the server.

To reduce this cost, the client can generate authentication tags on the data blocks of U (as discussed in Section~\ref{tag_gen})
and construct a Merkle tree over these tags instead of the data blocks. In this setting, $\Pi_{read}$ consists of $\tilde{h}(B)$,
the authentication tag in its sibling leaf-node and the hash values along the associated path.
This reduces the communication bandwidth between the client and the server for a read to $\beta+O(\lambda\log n)$.

\item \textbf{Write}\q
A write operation incurs the following costs.

\begin{itemize}
\item \textit{Write on U}:\q A write operation on U involves an authenticated read operation followed by
the verification of Eqn.~\ref{eqn:MHT}. Thus, each write operation requires $\beta+O(\lambda\log n)$ bandwidth
between the client and the server (for communicating $\Pi_{read}$ and $digMHT_{server}$).

\item \textit{Write on H and $\tilde{\text{H}}$}:\q The cost of a write on H is $O(\beta\log n)$ (see Section~\ref{buff_H}).
Similarly, the cost of a write on $\tilde{\text{H}}$ is $O(\lambda\log n)$ as the blocks are replaced
by their authentication tags in $\tilde{\text{H}}$ and
the size of a tag is $O(\lambda)$ bits.

\item \textit{Write on C and $\tilde{\text{C}}$}:\q C (or $\tilde{\text{C}}$) is rebuilt after every $n$ writes.
As mentioned in Section~\ref{buff_C}, a write operation on C costs $O(\beta\log n)$ both in time and bandwidth.
Similarly, the cost of a write on $\tilde{\text{C}}$ is $O(\lambda\log n)$.

\end{itemize}

\item \textbf{Audit}\q
For a challenge set $Q$ containing $r=O(\lambda\log n)$ random locations $\{\texttt{addr}_i\}_{1\le i\le r}$
and random coefficients $\nu_1,\nu_2,\ldots,\nu_r\in\Z_q^*$, the server computes a proof containing
$B^*=\sum_{1\le i\le r}{\nu_i}B_{\texttt{addr}_i}$ and $\{\tilde{h}(B_{\texttt{addr}_i})\}_{1\le i\le r}$
and sends the proof to the verifier.
Thus, the bandwidth required for an audit is given by
$\beta+O(\lambda^2\log n)$.

\end{itemize}

\medskip
\noindent
\textbf{Comparison among Dynamic POR Schemes}\q
We compare our scheme with other existing dynamic proofs-of-retrievability (POR) schemes which is summarized in Table~\ref{tab:comparison_POR}.
The comparison is based on the asymptotic complexity for different parameters.
Some of the figures mentioned in Table~\ref{tab:comparison_POR} are taken from~\cite{Stefanov_CCS}.
Table~\ref{tab:parameters} mentions typical values of the parameters used in our scheme~\cite{Krohn_SP}.

\begin{table*}[tbp]
\small
\centering
\caption{Comparison among dynamic POR schemes based on different parameters (asymptotic complexity)}\label{tab:comparison_POR}
\begin{tabular}{|c|c|c|c|c|c|c|}
\hline
Dynamic  & \multirow{2}{*}{Client} & \multicolumn{2}{c|}{Cost of a write operation} & \multicolumn{2}{c|}{Cost of an audit operation} & {\multirow{3}{*}{Verifiability}} \\
\cline{3-6}
POR & \multirow{2}{*}{storage} & Server  & \multirow{2}{*}{Bandwidth} & Server  & \multirow{2}{*}{Bandwidth} & \\
schemes &  & computation &  &  computation &  & \\
\hline
\hline
Iris~\cite{IRIS}			& $O(\beta\sqrt{n})$	& $O(\beta)$       		& $O(\beta)$					& $O(\beta\lambda\sqrt{n})$	& $O(\beta\lambda\sqrt{n})$	& Private\\
\hline
Cash et al.~\cite{Wichs_ORAM}		& $O(\beta)$		& $O(\beta\lambda(\log n)^2)$	& $O(\beta\lambda(\log n)^2)$			& $O(\beta\lambda(\log n)^2)$	& $O(\beta\lambda(\log n)^2)$	& Private\\
\hline
Chandran et al.~\cite{Bhavana_TCC}	& $O(\beta)$		& $O(\beta(\log n)^2)$		& $O(\beta(\log n)^2)$				& $O(\beta\lambda\log n)$	& $O(\beta\lambda\log n)$	& Private\\
\hline
\multicolumn{1}{|c|}{\multirow{2}{*}{Shi et al.~\cite{Stefanov_CCS}}}		& $O(\beta)$		& $O(\beta\log n)$		& $\beta+O(\lambda\log n)$			& $O(\beta\lambda\log n)$	& $\beta+O(\lambda^2\log n)$	& Private\\
\cline{2-7}
\multicolumn{1}{|c|}{}		& $O(\beta\lambda)$	& $O(\beta\log n)$		& $\beta(1+\epsilon)+O(\lambda\log n)^\dag$	& $O(\beta\lambda\log n)$	& $O(\beta\lambda\log n)$	& Public\\
\hline
Our scheme				& $O(\beta)$		& $O(\beta\log n)$		& $\beta+O(\lambda\log n)$			& $O(\beta\lambda\log n)$	& $\beta+O(\lambda^2\log n)$	& Public\\
\hline
\end{tabular}
\vspace{0.12in}
\begin{tablenotes}
\item[] We take $\lambda$ as the security parameter and $n$ as the number of blocks (each $\beta$-bits long)
	of the data file to be outsourced to the server. For all of the schemes mentioned above, the storage
	on the server side is $O(\beta n)$, where $\beta\gg\lambda$. The cost of an authenticated read operation
	is $\beta+O(\lambda\log n)$ if a Merkle hash tree is maintained over the unencoded data blocks
	for checking authenticity and freshness. \smallskip
\item[] $\dag$ $\epsilon$ is a constant such that $\epsilon>0$.
\end{tablenotes}
\end{table*}

From Table~\ref{tab:comparison_POR}, we note that, in our publicly verifiable dynamic POR scheme,
bandwidths required for a write and an audit are given by
$\beta+O(\lambda\log n)$ and $\beta+O(\lambda^2\log n)$, respectively. These figures are asymptotically
the same as those in the \textit{privately verifiable} scheme of~\cite{Stefanov_CCS}. On the other hand,
this is a significant improvement over the
\textit{publicly verifiable} scheme of~\cite{Stefanov_CCS}
where bandwidths required for a write and an audit are $\beta(1+\epsilon)+O(\lambda\log n)$ and $O(\beta\lambda\log n)$,
respectively, for a constant $\epsilon>0$ and $\beta\gg\lambda$.

Additionally, one drawback of the publicly verifiable scheme proposed by Shi et al.~\cite{Stefanov_CCS}
is due to the fact that one or more Merkle hash trees (\textit{separate} from
the Merkle hash tree\footnote{We note that, in our scheme as well as in~\cite{Stefanov_CCS},
a Merkle hash tree is maintained for the unencoded buffer U. However, as U is not audited
(its authenticity is checked only by the client during a read or write) by a third party auditor, the client keeps
the root-digest of this Merkle hash tree (for U) private avoiding frequent updates in the public parameters.} maintained for U)
are maintained to ensure the integrity of the blocks
in the \textit{hierarchical log} H (one for the entire log or one for each of its levels). To enable a third party
auditor (TPA) to audit the data blocks residing at different levels of this log, the root-digests of
these trees need to be made public. However, some of these root-digests are changed as often as new data blocks are inserted
in the hierarchical log structure, thus resulting in
a change in the public parameters for \textit{each} write. This incurs an additional (non-trivial) overhead for
validation and certification of the public parameters for every write operation. On the other hand, the public
parameters in our publicly verifiable dynamic POR scheme
are fixed throughout the execution of the protocols involved.

\begin{table}[t]
\small
\centering
\caption{Typical values of the parameters used}
\label{tab:parameters}
\begin{tabular}{|c|c|c|}
\hline
Parameter & Description of parameter & Value \\
\hline
\hline
$\lambda$		& Security parameter (in bits)		& 128 \\
\hline
$\lambda_p$		& Size of prime $p$ (in bits)		& 1024 \\
\hline
$\lambda_q$		& Size of prime $q$ (in bits)		& 257 \\
\hline
$\beta$			& Size of a data block (in KB)		& 64 \\
\hline
{\multirow{2}{*}{$m$}}	& $\ceil{\beta/(\lambda_q-1)}$  	& {\multirow{2}{*}{128}} \\
			& $=$ number of segments in a block 	&  \\
\hline
\end{tabular}
\end{table}

Apart from the schemes listed in Table~\ref{tab:comparison_POR}, we mention some POR schemes
proposed recently that handle data dynamics as follows.
The dynamic POR scheme proposed by Guan et al.~\cite{Guan_ESORICS}
uses the notion of indistinguishability obfuscation ($i\mathcal{O}$)~\cite{BarakIO_CR,GargIO_FOCS}
to construct a publicly verifiable POR scheme from the privately verifiable scheme of Shacham and Waters~\cite{SW_ACR}.
It also handles dynamic data using a ``modified B+ Merkle tree''.
However, the $i\mathcal{O}$ candidates available in the literature are not currently practical.
Ren et al.~\cite{Ren_TSC15} propose a dynamic POR scheme where the data file is encoded
using erasure coding (intra-server encoding) and network coding (inter-server encoding).
The encoded blocks are then disseminated among multiple storage servers.
Use of network coding reduces the communication bandwidth required for a repair in case of a node (server) failure.
For the intra-server encoding, each block is divided into some sub-blocks (using an erasure code),
and a ``range-based 2-3 tree'' (rb23Tree) is built upon these sub-blocks for each server.
This ensures the authenticity and freshness properties of the blocks within a server.
We note that each block is encoded (locally) into a few number of sub-blocks for the intra-server encoding.
Therefore, an update in a block (or in any of its sub-blocks) requires updating only a few sub-blocks corresponding to that block.
This makes an update in this scheme efficient. On the other hand, a malicious server needs to delete only a few sub-blocks
to actually make a block unavailable.
Thus, the dynamic POR scheme proposed by Ren et al.~\cite{Ren_TSC15} differs from our scheme on the basis of the granularity of data the client needs.

\section{Conclusion}
\label{sec:conclusion}

In this work, we have proposed a dynamic POR scheme where the client can update her data file
after the initial outsourcing of the file to the cloud server and retrieve all of her data at
any point of time. Our scheme is publicly verifiable, that is, anyone having the knowledge of
the public parameters of the scheme can perform an audit on the client's behalf, and
it offers security guarantees of a dynamic POR scheme. This scheme
is more efficient (in terms of the cost of a write or an audit) than other
practical and publicly verifiable dynamic POR schemes with a similar data granularity.

\end{document}